%% file: StableDGArxv.tex
\documentclass[11pt]{article}

\pdfoutput=1 

\usepackage{amsmath,amssymb,amsthm,color,graphicx}
\usepackage{euscript}
\usepackage{times}
\usepackage{mathpazo}

\usepackage[top=1.25in,bottom=1.25in,left=1.25in,right=1.25in]{geometry}

\usepackage{microtype}
\usepackage{fixltx2e}
\usepackage{fullpage}
\usepackage{xspace}

\newcommand{\lemlab}[1]{\label{lemma:#1}}

\newcommand{\lemref}[1]{Lemma~\ref{lemma:#1}}

\newtheorem{theorem}{Theorem}[section]

\newtheorem{lemma}[theorem]{Lemma}
\newtheorem{claim}[theorem]{Claim}
\graphicspath{{Figs/}}

\def\reals{{\mathbb R}}
\def \sphere{{\mathbb S}}

\def\G{\EuScript{G}}

\def\T{\EuScript{T}}

\def\bd{{\partial}}
\def\eps{{\varepsilon}}
\def\ph{{\varphi}}
\def\poly{\diamond}

\newcommand{\ignore}[1]{}

\def\bisect{b}
\def\Nbrs{N}

\def\distfn{\varphi}

\def\etal{\textsl{et~al.}}

\def\conv{\mathop{\mathrm{conv}}}

\def\SDG{\mathop{\mathrm{SDG}}}

\def\DT{\mathop{\mathrm{DT}}}
\def\VD{\mathop{\mathrm{VD}}}
\def\Vor{{\mathop{\mathrm{Vor}}}}
\def\intr{\mathop{\mathrm{int}}}

\makeatletter
\long\def\@makecaption#1#2{
   \vskip 10pt
   \setbox\@tempboxa\hbox{{\footnotesize \textbf{#1.} #2}}
   \ifdim \wd\@tempboxa >\hsize         
       {\footnotesize \textbf{#1.} #2\par}
     \else                              
       \hbox to\hsize{\hfil\box\@tempboxa\hfil}
   \fi}
\makeatother

\begin{document}

\begin{titlepage}

\title{
Stable Delaunay Graphs\thanks{%
An earlier version \cite{StableFull} of this paper appeared in
{\it Proc. 26th Annual Symposium on Computational Geometry},
2010, 127--136.}}
\author{Pankaj K. Agarwal\thanks{%
Department of Computer Science, Duke University, Durham, NC
27708-0129, USA, {\tt pankaj@cs.duke.edu}.}
\and
Jie Gao\thanks{%
Department of Computer Science, Stony Brook University, Stony
Brook, NY 11794, USA, {\tt jgao@cs.sunysb.edu}. }
\and
Leonidas J. Guibas\thanks{%
Department of Computer Science, Stanford University, Stanford,
CA 94305, USA, {\tt guibas@cs.stanford.edu}.}
\and
Haim Kaplan\thanks{%
School of Computer Science, Tel Aviv University, Tel~Aviv 69978, Israel.
{\tt haimk@tau.ac.il}.}
\and
Natan Rubin\thanks{%
Department of Computer Science, Ben-Gurion University of the Negev, Beer-Sheva, Israel 84105,
{\tt rubinnat.ac@gmail.com}.}
\and
Micha Sharir\thanks{%
School of Computer Science, Tel Aviv University, Tel~Aviv 69978, Israel;
and Courant Institute of Mathematical Sciences, New York University,
New York, NY~~10012,~USA.  {\tt michas@tau.ac.il}.}
}

\maketitle

\begin{abstract}
Let $P$ be a set of $n$ points in $\reals^2$, and let $\DT(P)$ denote its Euclidean Delaunay triangulation. We introduce the notion of an edge of $\DT(P)$ being \emph{stable}.
Defined in terms of a parameter $\alpha>0$, a 
Delaunay edge $pq$ is called $\alpha$-stable, if the (equal) angles at 
which $p$ and $q$ see the corresponding Voronoi edge $e_{pq}$ are at least $\alpha$.
A subgraph $G$ of $\DT(P)$ is called {\em $(c\alpha, \alpha)$-stable
Delaunay graph} ($\SDG$ in short), for some constant $c \ge 1$, 
if every edge in  $G$ is $\alpha$-stable and every $c\alpha$-stable of $\DT(P)$ is in $G$.

We show that if an edge is stable in the Euclidean Delaunay triangulation
of $P$, then it is also a stable edge, though for a different value of
$\alpha$, in the Delaunay triangulation of $P$ under any convex distance function that is 
sufficiently close to the Euclidean norm, and vice-versa. In particular, a 
$6\alpha$-stable edge in $\DT(P)$ is $\alpha$-stable in the Delaunay triangulation under the distance function induced by a regular $k$-gon for 
$k \ge 2\pi/\alpha$, and vice-versa. 
Exploiting this relationship and the analysis in~\cite{polydel},
we present a linear-size kinetic data
structure (KDS) for maintaining an $(8\alpha,\alpha)$-$\SDG$ as the 
points of $P$ move.  If the points move along
algebraic trajectories of bounded degree, the KDS
processes nearly quadratic events during the motion, each of which can processed in $O(\log n)$ time.
Finally, we show that a number of useful properties of $\DT(P)$ are retained by SDG of $P$.
\end{abstract}
\end{titlepage}

\section{Introduction}
\label{sec:intro}

Let $P$ be a set of $n$ points in $\reals^2$. For a point $p \in P$, the
(Euclidean) \emph{Voronoi cell} of $p$ is defined as 
$$
\Vor(p) = \{ x\in\reals^2 \mid \|xp\| \le \|xp'\|\, \forall p' \in P \}.
$$
The Voronoi cells of points in $P$ are nonempty, have pairwise-disjoint
interiors, and partition the plane. The planar
subdivision induced by these Voronoi cells is referred to as the 
(Euclidean) \emph{Voronoi diagram} of $P$ and we denote it as $\VD(P)$.
The \emph{Delaunay graph} of $P$ is the dual
graph of $\VD(P)$, i.e., $pq$ is an edge of the
Delaunay graph if and only if $\Vor(p)$ and $\Vor(q)$ share an edge.
This is equivalent to the existence of a circle passing through $p$ and $q$
that does not contain any other point of $P$ in its interior---any
circle centered at a point of $\bd\Vor(p)\cap\bd\Vor(q)$ and passing
through $p$ and $q$ is such a circle.  If no four
points of $P$ are cocircular, then the planar subdivision induced by the Delaunay
graph is a triangulation of the convex hull of $P$---the well-known (Euclidean) 
\emph{Delaunay triangulation} of $P$, denoted as $\DT(P)$.
See Figure~\ref{fig:VDT}~(a). 
$\DT(P)$ consists of all triangles whose circumcircles do not
contain points of $P$ in their interior. 
Delaunay triangulations and Voronoi diagrams are fundamental to much of
computational geometry and its applications. See \cite{AKL} for a
very recent textbook on these structures.


\begin{figure}[htbp]
\centering
\begin{tabular}{ccc}
\includegraphics[scale=1.0]{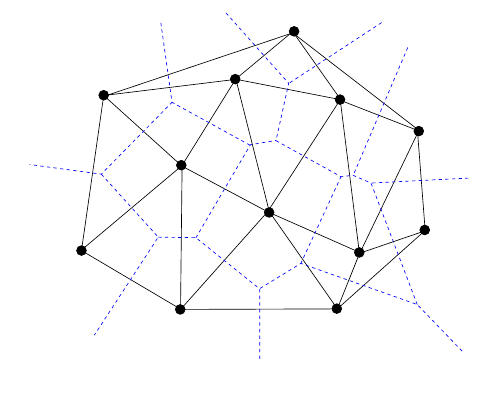}&\hspace*{1cm}&
\includegraphics[scale=1.0]{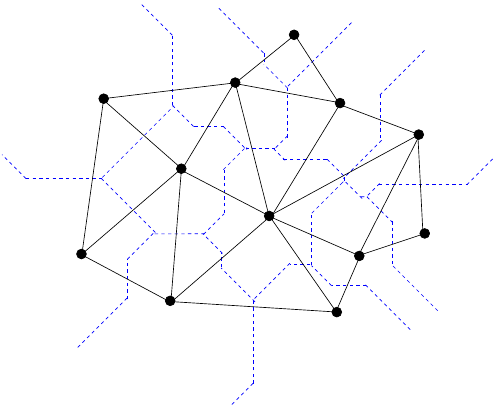}\\
\small (a) && \small (b)
\end{tabular}
\caption{(a) The Euclidean Voronoi diagram (dotted) and Delaunay
triangulation (solid); (b) the $Q$-Voronoi diagram and $Q$-Delaunay triangulation for an axis-parallel square
$Q$, i.e., the Voronoi and Delaunay diagrams under the $L_\infty$-metric.}
\label{fig:VDT}
\end{figure}
In many applications of Delaunay/Voronoi methods (e.g., mesh generation
and kinetic collision detection), the input points are moving continuously,
so they need to be efficiently updated as motion occurs.
Even though the motion of the points is continuous, the combinatorial and
topological structure of $\VD(P)$ and $\DT(P)$ change only at
discrete times when certain ``critical events'' occur. A challenging open question
in combinatorial geometry is to bound the number of critical events if each point
of $P$ moves along an algebraic trajectory of constant degree.

Guibas et al.~\cite{gmr-vdmpp-92} showed a roughly cubic upper bound of
$O(n^2 \lambda_s(n))$ on the number of critical events.
Here $\lambda_s(n)$ is the maximum length
of an $(n,s)$-Davenport-Schinzel sequence~\cite{SA95}, and $s$ is a constant
depending on the degree of the motion of the points.
See also Fu and Lee~\cite{FL}. The best known lower
bound is quadratic~\cite{SA95}. Recent works of Rubin \cite{Rubin, RubinUnit} 
establish an almost quadratic bound of $O(n^{2+\eps})$, for any
$\eps>0$, for the restricted cases where any four points of $P$ can
be cocircular at most {\it two} or {\it three} times. In particular,
the latter study \cite{RubinUnit} covers the case of points moving
along lines at common unit speed, which has been highlighted as a major
open problem in discrete and computational geometry; see
\cite{TOPP}. Nevertheless, no sub-cubic upper bound is known for
more general motions, including  the case where the points of $P$
are moving along lines at non-uniform, albeit fixed, speeds. It is
worth mentioning that the analysis in \cite{Rubin}, and even more so
in \cite{RubinUnit}, is fairly involved, which results in a huge
implicit constant of proportionality.

Given this gap in the bound on the number of critical events, it is natural to ask whether
one can define a large subgraph of the Delaunay graph of $P$ so that
(i) it provably experiences at most a nearly quadratic number of
critical events, (ii) it is reasonably easy to define and maintain,
and (iii) it retains useful properties for further applications. 
This paper defines such a subgraph of the Delaunay graph, shows 
that it can be maintained efficiently, and proves that it preserves 
a number of useful properties of $\DT(P)$.

\paragraph{Related work.}
It is well known that $\DT(P)$ can be maintained efficiently using the so-called
kinetic data structure framework proposed by Basch~\etal~\cite{bgh-dsmd-99}. 
A triangulation $\T$ of the convex hull $\conv(P)$ of $P$ is the
Delaunay triangulation of $P$
if and only if for every edge $pq$ adjacent to two triangles
$\triangle pqr^+$ and $\triangle pqr^-$ in $\T$, the circumcircle of
$\triangle pqr^+$ (resp.,\ $\triangle pqr^-$) does not contain $r^-$
(resp.,\ $r^+$).  Equivalently,
\begin{equation}\label{eq:DT}
\angle pr^+q + \angle pr^-q < \pi .
\end{equation}
Equality occurs when $p,q,r^+,r^-$ are cocircular, which generally
signifies that a combinatorial change in $\DT(P)$ (a so-called 
\emph{edge flip}) is about to take place.
We also extend (\ref{eq:DT}) to apply to edges $pq$ of the hull, each having
only one adjacent triangle, $\triangle pqr^+$. In this case we take $r^-$
to lie at infinity, and put $\angle pr^-q=0$. An equality in (\ref{eq:DT})
occurs when $p,q,r^+$ become collinear (along the hull boundary), and again
this signifies a combinatorial change in $\DT(P)$.

This makes the maintenance of $\DT(P)$ under point motion quite simple: an
update is necessary only when the empty circumcircle condition (\ref{eq:DT})
fails for one of the edges, i.e., for an edge $pq$, adjacent to triangles
$\triangle pqr^+$ and $\triangle pqr^-$, $p$, $q$,
$r^+$, and $r^-$ become cocircular.\footnote{We assume the motion of the
points to be sufficiently generic, so that no more than four points can
become cocircular at any given time, and so that equality in (\ref{eq:DT})
is not a local maximum of the left-hand side.}
Whenever such an event happens,
the edge $pq$ is \emph{flipped} with $r^+r^-$ to restore Delaunayhood. 
Keeping track of these cocircularity events is straightforward, and each
such event is detected and processed in $O(\log n)$ time
\cite{gmr-vdmpp-92}. However, as
mentioned above, the best known upper bound on the number of events processed by this
KDS (which is the number of topological changes in $\DT(P)$ during the motion), 
assuming that the points of $P$ are moving along algebraic trajectories of
bounded degree, is near cubic~\cite{gmr-vdmpp-92} (except for the special 
cases treated in \cite{Rubin,RubinUnit}).

So far we have only considered the Euclidean Voronoi and Delaunay diagrams,
but a considerable amount of literature exists on Voronoi and Delaunay diagrams
under other norms and so-called \emph{convex distance functions}; see
Section~\ref{sec:prelim} for details.

Chew~\cite{Chew} showed that the Delaunay triangulation of $P$ under the $L_1$- or
$L_\infty$-metric experiences only a near-quadratic number of events, if 
the motion of the points of $P$ is
algebraic of bounded degree. In the companion paper~\cite{polydel}, we present a
kinetic data structure for maintaining the Voronoi diagram and Delaunay
triangulation of $P$ under a polygonal convex distance function for an 
arbitrary convex polygon $Q$.
(see Section~\ref{sec:prelim} for the definition) that processes only a
near-quadratic number of events,
and can be updated in $O(\log n)$ time at each event. Since a regular convex
$k$-gon approximates a circular disk, it is tempting to maintain the 
Delaunay triangulation under a polygonal convex distance function as a 
(hopefully substantial) portion of the Euclidean Delaunay graph of $P$.
Unfortunately, the former is not necessarily a subgraph of the latter~\cite{AKL}.

Many subgraphs of $\DT(P)$, such as the Euclidean minimum spanning tree (MST), 
Gabriel graph, relative neighborhood graph, and $\alpha$-shapes, have been used
extensively in a wide range of applications (see e.g.~\cite{AKL}). However no
sub-cubic bound is known on the number of discrete changes in their structures
under an algebraic motion of the points of $P$ of bounded degree. Furthermore, no efficient
kinetic data structures are known for maintaining them, for unlike $\DT(P)$, they
may undergo a ``non-local'' change at a critical event; see~\cite{AEGH, BGZProximity, Proximity} 
for some partial results on maintaining an MST.

\paragraph{Our results.}
\textit{\textbf{Stable Delaunay edges:}}\hspace{3mm}
We introduce the notion of \emph{$\alpha$-stable Delaunay edges},
for a fixed parameter $\alpha>0$, defined as follows.  Let $pq$ be
a Delaunay edge under the Euclidean norm, and let $\triangle pqr^+$
and $\triangle pqr^-$ be the two Delaunay triangles incident to $pq$.
Then $pq$ is called {\em $\alpha$-stable} if its opposite angles in
these triangles satisfy
\begin{equation}
\label{eq:SDG}
\angle pr^+q + \angle pr^-q \leq \pi-\alpha  .
\end{equation}
As above, the case where $pq$ lies on $\conv(P)$ is treated as if
 $r^-$, say,  lies at infinity, so that the corresponding
angle  $\angle pr^-q$ is equal to $0$. An equivalent and more useful
definition, in terms of the dual Voronoi diagram, is that $pq$ is
$\alpha$-stable if \textit{the angles at which $p$ and $q$ see
their common Voronoi edge $e_{pq}$ are at least $\alpha$ each.}
See Figure~\ref{Fig:LongDelaunay}(a).

In the case where $pq$ lies on $\conv(P)$ the corresponding
dual Voronoi edge $e_{pq}$ is an infinite ray emanating from some
Voronoi vertex $x$. We define the angle in which a point $p$ {\em
see} such a Voronoi ray to be the angle between the segment $px$ and
an infinite ray parallel to $e_{pq}$ emanating from $p$. With this
definition of the angle in which a point of $\conv(P)$  sees a
Voronoi ray, it is easy to check that the alternative definition of
$\alpha$-stability is equivalent to the ordinal definition
also when $pq$ lies on $\conv(P)$.
We call the Voronoi edges corresponding to $\alpha$-stable Delaunay edges
\emph{$\alpha$-long} (and call the remaining edges \emph{$\alpha$-short}). 
See Figure~\ref{Fig:LongDelaunay}. Note that
for $\alpha=0$, when no four points are cocircular, (\ref{eq:SDG})
coincides with (\ref{eq:DT}).

\begin{figure}[htbp]
\centering
\begin{tabular}{ccc}
\input{LongDelaunay.pspdftex}&\hspace*{15mm}&
\includegraphics[scale=1.0]{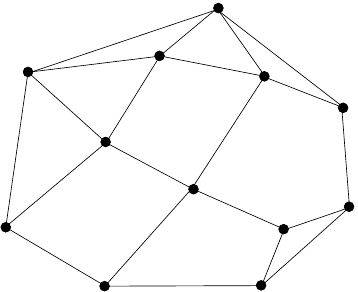}\\
\small (a)&&\small(b)
\end{tabular}
\caption{(a) An $\alpha$-stable edge $pq$ in $\DT(P)$, and its dual edge
$e_{pq}=ab$ in $\VD(P)$; the fact that $\angle r^+ + \angle r^- = \pi-\alpha$
follows by elementary geometric considerations. 
(b) A SDG of the Delaunay triangulation $\DT(P)$ shown in 
Figure~\protect\ref{fig:VDT}~(a), with $\alpha=\pi/8$.}
\label{Fig:LongDelaunay}
\end{figure}

A justification for calling such edges stable lies in the following
observation: If a Delaunay edge $pq$ is $\alpha$-stable then it
remains in $\DT(P)$ during any continuous motion of the points of $P$ for which
every angle $\angle prq$, for $r\in P\setminus\{p,q\}$, changes
by at most $\alpha/2$. This is clear because, as is easily verified,
at any time when $pq$ is $\alpha$-stable we have
$\angle pr^+q + \angle pr^-q \leq \pi-\alpha$ for \emph{any} pair
of points $r^+$, $r^-$ lying on opposite sides of the line $\ell$
supporting $pq$, so,
if each of these angles changes by at most $\alpha/2$ we still have
$\angle pr^+q + \angle pr^-q \le \pi$ for every such pair $r^+$, $r^-$,
implying that $pq$ remains an edge of $\DT(P)$.\footnote{%
  This argument also covers the cases when a point $r$ crosses $\ell$
  from side to side: Since each point, on either side of $\ell$, sees
  $pq$ at an angle of $\leq \pi-\alpha$, it follows that no point can cross
  $pq$ itself -- the angle has to increase from $\pi-\alpha$ to $\pi$. Any
  other crossing of $\ell$ by a point $r$ causes $\angle prq$ to
  decrease to $0$, and even if it increases to $\alpha/2$ on the other side
  of $\ell$, $pq$ is still an edge of $\DT$, as is easily checked.}
Hence, as long as the ``small angle change'' condition
holds, stable Delaunay edges remain a ``long time'' in the
triangulation.  Informally speaking, the non-stable edges $pq$
of $\DT(P)$ are those for which $p$ and $q$
are almost cocircular with their two common Delaunay neighbors
$r^+$, $r^-$, and hence $pq$ is more likely to get flipped ``soon.''

\medskip

\textit{\textbf{Stable Delaunay graph:}}\hspace{3mm}
For two parameters $0 \le \alpha \le \alpha' \le \pi$,
we call a subgraph $\G$ of $\DT(P)$
an {\em $(\alpha',\alpha)$-stable Delaunay graph} (an
$(\alpha',\alpha)$-$\SDG$ for short) if
\begin{itemize}
\item[(S1)] every edge of $\G$ is $\alpha$-stable, and
\item[(S2)] every $\alpha'$-stable edge of $\DT(P)$ belongs to $\G$.
\end{itemize}
Note that an $(\alpha',\alpha)$-$\SDG$ is not uniquely defined even
for fixed $\alpha, \alpha'$ because the edges that are $\alpha$-stable
but not $\alpha'$-stable may or may not be in $\G$. Throughout this paper,
$\alpha'$ will be some fixed (and reasonably small) multiple of $\alpha$.

\begin{figure}[hbt]
\begin{center}
\input{normasis.pspdftex}
\caption{A convex distance function, induced by $Q$, that is $\alpha$-close
to the Euclidean norm.}
\label{fig:normasis}
\end{center}
\end{figure}
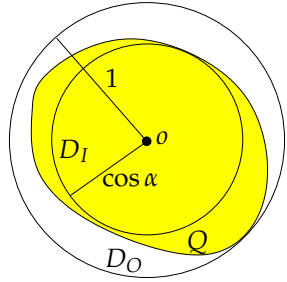

Our main result is that a stable edge of the Euclidean Delaunay triangulation appears
as stable edge in the Delaunay triangulation under any convex distance function
that is sufficiently close to the Euclidean norm (see Section~\ref{sec:prelim} for more details). 
More precisely, we say that the  distance function induced by a compact convex set $Q$ is
\emph{$\alpha$-close} to the Euclidean norm if $Q$ is contained in
the unit disk $D_O$ and contains the disk $D_I = (\cos\alpha)
D_O$ both centered in the origin.\footnote{%
The Hausdorff distance between $Q$ and $D_O$ is at most
$1-\cos\alpha\approx \alpha^2/2$.}
See Figure \ref{fig:normasis}.
In particular, for $k=\pi/\alpha$, the regular $k$-gon $Q_k$ is such a set, as easy trigonometry shows. 
We prove the following:
\begin{theorem} \label{thm:Qnorm}
Let $P$ be a set of $n$ points in $\reals^2$, $\alpha>0$ a parameter,
and $Q$ a compact, convex set inducing a convex distance
function $d_Q(\cdot,\cdot)$ that is $\alpha$-close to the Euclidean norm. Then the following properties hold.
\begin{itemize}
\item[(i)] Every $11\alpha$-stable Delaunay edge under the Euclidean norm
is an $\alpha$-stable Delaunay edge under $d_Q$.
\item[(ii)] Symmetrically, every $11\alpha$-stable Delaunay edge under
$d_Q$ is also an $\alpha$-stable Delaunay edge under the Euclidean norm.
\end{itemize}
\end{theorem}

In particular, if $Q$ is a regular $k$-gon for $k \ge 2\pi/\alpha$, then the
above theorem holds for $Q$. In the companion paper~\cite{polydel}, we have 
presented an efficient kinetic data structure for maintaining the Delaunay 
triangulation and Voronoi diagram of $P$ under a polygonal convex distance function. 
Using this result, we obtain the second main result of the paper:

\begin{theorem}\label{Thm:MaintainSDGPolyg}
Let $P$ be a set of $n$ moving points in $\reals^2$ under algebraic
motion of bounded degree, and let $\alpha \in (0,\pi)$ be a parameter.
A Euclidean $(8\alpha,\alpha)$-stable Delaunay graph of $P$ can be
maintained by a linear-size KDS that processes
$O\left(\frac{1}{\alpha^4}n\lambda_r(n)\right)$ events and
updates the SDG at each event in $O(\log n)$ time.
 Here $r$ is a constant that depends on the degree of the motion of the points
of $P$, and $\lambda_r(n)$ is the maximum length of a Davenport-Schnizel sequence of order $r$.
\end{theorem}

For simplicity, we first prove in Section~\ref{sec:polysdg}
Theorem~\ref{thm:Qnorm} for the case when $Q$ is a regular $k$-gon, for $k \ge 2\pi/\alpha$, and use the
argument to prove Theorem~\ref{Thm:MaintainSDGPolyg}. 
Actually, we prove Theorem~\ref{thm:Qnorm} with a slightly better
constant using the additional structure possessed by the diagrams when 
$Q$ is a regular $k$-gon. Next, we prove
in Section~\ref{sec:Qnorm} Theorem~\ref{thm:Qnorm} for an arbitrary $Q$.
Finally, we prove in Section~\ref{sec:SDGProperties} a few useful properties of 
$\DT(P)$ that are retained by the stable Delaunay graph of $P$.

\section{Preliminaries}
\label{sec:prelim}

This section introduces a few notations, definitions, and known results that we 
will need in the paper.

We represent a direction in $\reals^2$ as a point on the unit circle $\sphere^1$. 
For a direction $u \in \sphere^1$ and an angle $\theta \in [0,2\pi)$, we use $u+\theta$ 
to denote the direction obtained after rotating the vector $\vec{ou}$ by angle $\theta$ in clockwise direction.
For a point $x \in \reals^2$ and a direction $u \in \sphere^1$, let $u[x]$ denote the ray emanating from $x$ in direction $u$.

\paragraph{$Q$-distance function.} 
Let $Q$ be a compact, convex set with non-empty interior and with the origin, denoted by $o$, lying in its interior. A
homothetic copy $Q'$ of $Q$ can be represented by a pair
$(p,\lambda)$, with the interpretation $Q' = p + \lambda Q$; $p$ is
the \emph{placement} (location) of the center $o$ of $Q'$, and
$\lambda$ is its {\em scaling factor} (about its center). $Q$ defines a distance function (also called the \emph{gauge} of $Q$)
$$
d_Q(x,y)=\min \{\lambda\mid y\in x+\lambda Q\}.
$$
Note that, unless $Q$ is centrally symmetric with respect to the origin, $d_Q$ is not symmetric.

Given a finite point set $P\subset\reals^2$ and a point $p\in P$,
we denote by $\Vor^Q(p)$, $\VD^Q(P)$, and $\DT^Q(P)$ the Voronoi cell of $p$,
the Voronoi diagram of $P$, and the Delaunay triangulation of $P$, respectively,
under the distance function $d_Q(\cdot,\cdot)$; see Figure~\ref{fig:VDT}~(b). 
To be precise (because of the potential asymmetry of $d_Q$), we define
$$
\Vor(p) = \{ x\in\reals^2 \mid d_Q(x,p) \le d_Q(x,p') \, \forall p' \in P \},
$$
and then $\VD^Q(P)$ and $\DT^Q(P)$ are defined in complete analogy to the
Euclidean case.  We refer the reader to~\cite{polydel} for formal definitions and 
details of these structures. Throughout this paper, we will drop the superscript 
$Q$ from $\Vor^Q, \VD^Q, \DT^Q$ when referring to them under the Euclidean norm.


\begin{figure}[htb]
\begin{center}
\input{Qcopies.pspdftex}\hspace{2.5cm}\input{TopologicQ.pspdftex}
\caption{(a) Homothetic copies $Q_{p}(u)$, centered at $c\in
u[p]$ (dashed), and $Q_{pq}(u)$ (solid) of $Q$. The ray $u[p]$ hits
$b^Q_{pq}$ at the center $c_{pq}$ of $Q_{pq}(u)$, which is defined,
when $Q$ is smooth at its contact with $p$,
if and only if $q$ and $Q_p(u)$ lie on the same side of
$\ell_{p}(u)$; (b) Sliding $Q_{pr}(u)$ away from $q$.}
\label{Fig:Qcopies}
\end{center}
\end{figure}
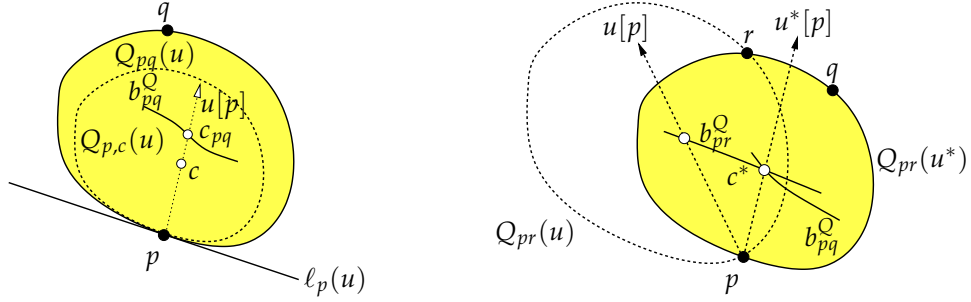

For a point $z\in\reals^2$, let $Q[z]$ denote the homothetic copy of $Q$
centered at $z$ such that its boundary touches the
$Q$-nearest neighbor(s) of $z$ in $P$, i.e., $Q[z]$ is
represented by the pair $(z,\lambda)$ where
$\lambda=\min_{p\in P} d_Q(z,p)$.
In other words, $Q[z]$ is the largest homothetic copy of $Q$ that is
centered at $z$ whose interior is $P$-empty.
We also use the notation $Q_p(u)$ to denote a ``generic''
homothetic copy of $Q$ which touches $p$ and is centered at some
point on $u[p]$. See Figure \ref{Fig:Qcopies}~(a). Note that all
homothetic copies of $Q_p(u)$  touch $p$ at the same point $\zeta$ of $\bd Q$,
and therefore share the same tangent at $p$. This tangent is unique if $Q$ is
smooth at $\zeta$, and we denote it by $\ell_p(u)$. If $Q$ is not smooth at $\zeta$
then there is a nontrivial range of tangents (i.e., supporting lines) to $Q$ at
$\zeta$; we can take $\ell_p(u)$ to be any of them, and it will be a supporting line
to all the copies $Q_p(u)$ of $Q$.

For a pair of points $p,q\in P$, let $b_{pq}^Q$ denote the
\emph{$Q$-bisector} of $p$ and $q$---the locus of all
placements of the center of any homothetic copy $Q'$ of $Q$ that
touches $p$ and $q$. If $Q$ is strictly convex or if $Q$ is not strictly convex but 
no two points are collinear with a straight segment on $\bd Q$, then $b_{pq}^Q$ is a 
one-dimensional curve and any ray $u[p]$ that hits $b_{pq}^Q$ does so at a unique point.
For such a direction $u$ and a pair of points $p,q \in P$, let $Q_{pq}(u)$ denote the homothetic 
copy of $Q$ that touches $p$ and $q$, whose center is $u[p]\cap b_{pq}^Q$.

If $Q$ is not strictly convex and $p,q$ are points in $P$ such that $\vec{pq}$
is parallel to a straight portion $e$ of $\bd Q$ then $b_{pq}^Q$ is not one-dimensional. 
In this case $Q_{pq}(u)$ is not well defined when $u$ is a direction that connects $e$ to the center $c$ of $Q$.
As is easy to check, in any other case the ray $u[p]$ either hits $b_{pq}$ at a unique point which determines $Q_{pq}(u)$, or entirely misses $b_{pq}^Q$. See the companion paper~\cite{polydel} for a detailed discussion of this phenomenon.

A useful property of the $Q$-bisectors is that any two bisectors
$b^{Q}_{pq},b^{Q}_{pr}$ with a common generating point $p$,
intersect exactly once, namely, at the center $c^*$ of the unique
homothetic copy of $Q$ that simultaneously touches $p,q$ and $r$
\cite{LS}. For this property to hold, though, we need to assume
that (i) the points $p,q,r$ are not collinear, and (ii) either $Q$ is strictly convex, or, otherwise, that none of the directions $\vec{pq}, \vec{pr}$ 
is parallel to a straight portion of
$\bd Q$. 
(The precise condition is that $b_{pq}$ and $b_{pr}$ be one-dimensional in a neighborhood of $c^*$.)
The local topology of the restricted $Q$-Voronoi diagram
$\VD^Q(\{p,q,r\})$ near $c^*$ is largely determined by the
orientation of the triangle $\triangle pqr$. Specifically, assume
with no loss of generality that $\vec{pr}$ is counterclockwise to
$\vec{pq}$, and let $u^*$ be the direction of the ray $\vec{pc^*}$.
Refer to Figure \ref{Fig:Qcopies}~(b). 
If we continuously rotate
a ray $u[p]$, for $u\in \sphere^1$, in counterclockwise
direction from $u^*[p]$, the corresponding copy
$Q_{pr}(u)$ will slide away from its contact with $q$
because the portion of $Q_{pr}(u)$ to the right of $\vec{pr}$ shrinks during the rotation.
Therefore, the rotating ray $u[p]$ either
misses $b^Q_{pq}$ entirely or hits $b^Q_{pq}$ after $b^Q_{pr}$. A
symmetric phenomenon, with $q$ and $r$ interchanged, takes place if
we rotate the ray $u[p]$ in clockwise direction from $u^*[p]$. 

It is known that $\Vor^Q(p)$, for every point $p \in P$, is star-shaped~\cite{AKL}, which 
implies that each Voronoi edge $e^Q_{pq}$ is fully enclosed
between the two rays that emanate from $p$, or from $q$, through its
endpoints. We remark that, unlike the Euclidean case, the angles
$\angle xpy$, $\angle xqy$ need not be equal in general.

Finally, we extend the notion of stable edges to $\DT^Q(P)$. 
We call an edge $pq \in \DT^Q(P)$
\emph{$\alpha$-stable} if the following property holds for
the dual edge $e_{pq}^Q$ in $\VD^Q(P)$:
\begin{quote}
\textit{Each of the points $p,q$ sees their common $Q$-Voronoi edge
$e_{pq}^Q$ at angle at least $\alpha$. That is, if $x$ and $y$
are the endpoints of $e_{pq}^Q$, then 
$\min\{\angle xpy, \angle xqy\} \ge \alpha$.}
\end{quote}
This definition coincides with the definition of
$\alpha$-stability under the Euclidean norm when $Q$ is the unit disk
(and in this case both angles are equal).

\paragraph{Remark.} If $Q$ is not strictly convex (and $pq$ is parallel to a straight portion of $\partial Q$), the endpoints of $e_{pq}^Q$ may be not well defined.
In this case, we resort to the following, more careful definition of $\alpha$-stability.
A ray $u[p]$ is said to (properly) {\it cross} $e_{pq}^Q$ only if the copy $Q_{pq}(u)$ is uniquely defined. The center of such a copy $Q_{pq}(u)$ necessarily lies within the one-dimensional portion $\tilde{e}_{pq}^Q$ of $Q_{pq}(u)$, which is easily seen to be non-empty and connected. 
We say that the edge $e_{pq}^Q$ is $\alpha$-stable if the set of rays $u[p]$ properly crossing $e_{pq}^Q$ spans an angle of at least $\alpha$, and a symmetric condition holds for the rays emanating from $q$.
In other words, our notion of $\alpha$-stability ignores the two-dimensional regions of $e_{pq}^Q$ (if these exist).\footnote{As is easy to check, the one-dimensional portion $\tilde{e}_{pq}^Q$ of $e^Q_{pq}$ varies continuously (in Hausdorff sense) with any sufficiently small perturbation of $p$ and $q$ within $P$. Furthermore, it is the {\it only} such portion: If a ray $u[p]$ hits $e^Q_{pq}$ outside $\tilde{e}_{pq}^Q$ (i.e., within its two-dimensional portion), there is a symbolic perturbation of $p$ and $q$ causing $u[p]$ to completely miss $e_{pq}^Q$.}

\paragraph{Polygonal convex distance function.} 
As mentioned in the introduction, we will be considering the case when $Q$ is a regular $k$-gon, for some 
even integer $k \ge 2\pi/\alpha$, centered at the origin. Let 
$v_0,\ldots,v_{k-1}$ be its sequence of vertices arranged in clockwise direction. 
For each $0\leq j< k$, let $u_j$ be the direction of the vector that 
connects $v_j$ to the center of $Q$ (see Figure~\ref{Fig:bisect}~(a)).  
We will use $b^\poly_{pq}, \Vor^\poly(p), \VD^\poly(P), \DT^\poly(P)$ to 
denote $b^Q_{pq}, \Vor^Q(p), \VD^Q(P), \DT^Q(P)$, respectively, when $Q$ 
is a regular $k$-gon. 

We say that $P$ is in \emph{general position} (with respect to $Q$) if no
three points of $P$ lie on a line, no two points of $P$ lie on a
line parallel to an edge or a diagonal of $Q$, and no four points of
$P$ are \emph{$Q$-cocircular}, i.e., no four points of $P$ lie on
the boundary of a common homothetic copy of $Q$.

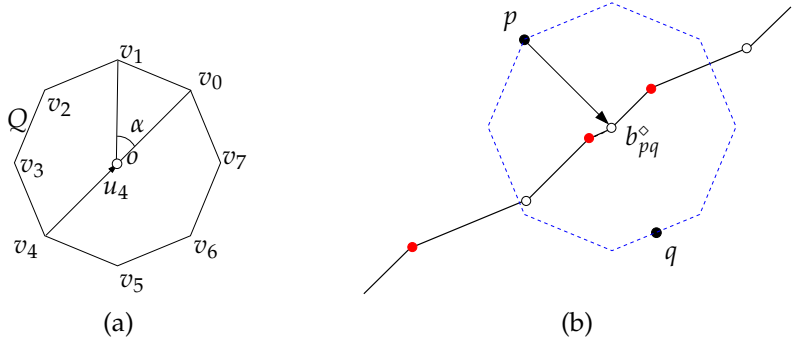
\begin{figure}[htb]
\centering
\begin{tabular}{ccc}
\input{convexkgon.pspdftex}&\hspace*{10mm}&
\input{bisect.pspdftex}\\
\small(a)&&\small(b)
\end{tabular}
\caption{(a) A regular octagon $Q$ centered at the origin. 
(b) The bisector $\bisect_{pq}^\poly$ for the
regular octagon $Q$; it has six breakpoints and the corner contacts
along $\bisect_{pq}^\poly$ alternate between contacts at $p$ (hollow
circles) and contacts at  $q$ (filled circles).} 
\label{Fig:bisect}
\end{figure}

The placements on $\bisect_{pq}^\poly$ at which
(at least) one of $p$ and $q$, say, $p$, touches $Q'$ at
a vertex is called a \emph{corner placement} (or
a \emph{corner contact}) at $p$; see Figure~\ref{Fig:bisect}~(b).
We also refer to these points on $\bisect_{pq}^\poly$ as
\emph{breakpoints}. 
We call a homothetic copy of $Q$ whose vertex $v_j$ touches a point $p$, a
{\em $v_j$-placement of $Q$ at $p$}. 

The following property of $\bisect_{pq}^\poly$ is proved in~\cite[Lemma~2.5]{polydel}:
\begin{lemma} \lemlab{interleave}
Let $Q$ be a regular $k$-gon, and let $p$ and $q$ be two points
in general position with respect to $Q$. 
Then $b^\poly_{pq}$ is a polygonal chain with $k-2$ breakpoints and the 
breakpoints along $b_{pq}^\poly$ alternate between corner contacts at $p$ and corner contacts at $q$.
\end{lemma}


\section{$\mathbf{DT^\poly(P)}$ and Euclidean $\SDG$'s}
\label{sec:polysdg}

In this section we first prove a slightly stronger version of 
Theorem~\ref{thm:Qnorm} for the case when $Q$ is a regular $k$-gon for some even
integer $k\ge 2\pi/\alpha$, and then prove Theorem~\ref{Thm:MaintainSDGPolyg}.
We follow the notation in Section~\ref{sec:prelim}, and, for simplicity, we assume that $\alpha=2\pi/k$.

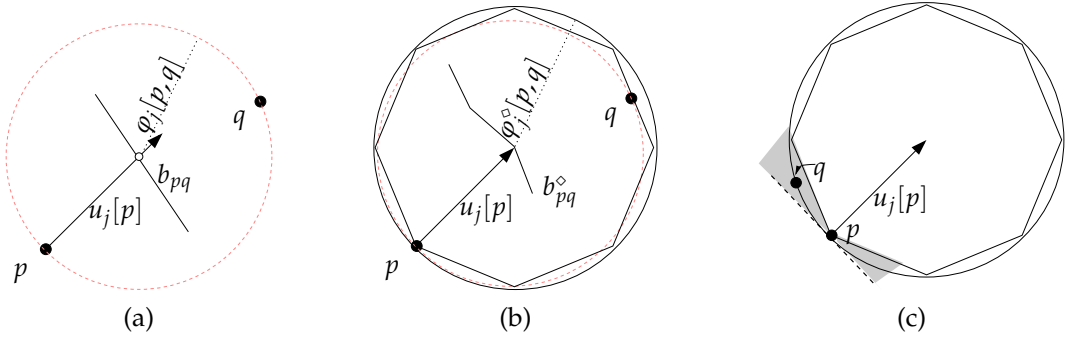
\begin{figure}[htb]
\centering
\begin{tabular}{ccccc}
\input{Euclphi.pspdftex}&\hspace*{5mm}&
\input{PolygonalDist.pspdftex}&\hspace*{5mm}&
\input{UndefinedPolygDist.pspdftex}\\
\small (a) && \small (b)&& \small (c)
\end{tabular}
\caption{(a)  The function $\ph_j[p,q]$, which is equal to the radius of the circle that pass through $p$ and $q$ and whose center lies on $u_j[p]$. 
(b) The bisector $\bisect^\poly_{pq}$ and the function
$\varphi_j^\poly[p,q]$, which is equal to the radius of the circle
that circumscribes the $v_j$-placement of $Q$ at $p$ that also
touches $q$. (c) The case when $\varphi^\poly_j[p,q]=\infty$ while
$\varphi_j[p,q]<\infty$. In this case $q$ lies in one of the shaded
wedges.} 
\label{Fig:Placement}
\end{figure}

For any pair $p,q\in P$, let $\ph_j[p,q]$ denote the distance
from $p$ to the point $u_j[p]\cap \bisect_{pq}$; we put
$\varphi_j[p,q]=\infty$ if $u_j[p]$ does not intersect
$\bisect_{pq}$. See Figure~\ref{Fig:Placement}~(a).
The point $q$ minimizes $\distfn_i[p,q']$,
among all points $q'$ for which $u_i[p]$ intersects $\bisect_{pq'}$,
if and only if the intersection between $\bisect_{pq}$ and $u_i[p]$
lies on the Voronoi edge $e_{pq}$. We call $q$ the
{\em neighbor of $p$ in direction $u_i$}, and denote it by $\Nbrs_i(p)$.

Similarly, let $\varphi^\poly_j[p,q]$ denote the distance
from $p$ to the point $u_j[p]\cap \bisect^\poly_{pq}$; we put
$\varphi^\poly_j[p,q]=\infty$ if $u_j[p]$ does not intersect
$\bisect^\poly_{pq}$. If $\varphi^\poly_j[p,q]<\infty$ then the point
$\bisect^\poly_{pq}\cap u_j[p]$ is the center of the $v_j$-placement
$Q'$ of $Q$ at $p$ that also touches $q$, and there is a unique such point.
The value $\varphi^\poly_j[p,q]$ is equal to the circumradius of $Q'$.
See Figure \ref{Fig:Placement}~(b).
The \emph{neighbor} $\Nbrs^\poly_j[p]$ of $p$ in direction $u_j$ is defined
to be the point $q\in P\setminus\{p\}$ that minimizes $\varphi^\poly_j[p,q]$.

Note that $\varphi_j[p,q]<\infty$ if and only if the angle between $\vec{pq}$ and $u_j[p]$ is 
smaller than $\pi/2$.  In contrast, $\varphi^\poly_j[p,q]<\infty$ if and only if
the angle between $\vec{pq}$ and $u_j[p]$ is at most
$\pi/2-\pi/k=\pi/2-\alpha/2$. Moreover, we have $\varphi_j[p,q]\leq \varphi^\poly_j[p,q]$ 
(see Figure~\ref{Fig:Placement}). Therefore, $\varphi^\poly_j[p,q]<\infty$ always
implies $\varphi_j[p,q]<\infty$, but not vice versa; see
Figure~\ref{Fig:Placement}~(c). Note also that in both the Euclidean and
the polygonal cases, the respective quantities $N_j[p]$ and $N_j^\poly[p]$
may be undefined.

\begin{lemma}\label{lemma:LongEucPoly}
Let $p,q\in P$ be a pair of points such that $\Nbrs_j(p)=q$ for
$h\geq 3$ consecutive indices, say $0\leq j\le h-1$.
Then for each of these indices, except possibly for the first and
the last one, we also have $\Nbrs^\poly_j[p]=q$.
\end{lemma}

\begin{proof}
Let $w_1$ (resp., $w_2$) be the point at which the ray $u_0[p]$
(resp., $u_{h-1}[p]$) hits the edge $e_{pq}$ in $\VD(P)$.
(By assumption, both points exist.) Let $D_1$ and $D_2$ be the disks
centered at $w_1$ and $w_2$, respectively, and touching $p$ and $q$.
By definition, neither of these disks contains a point of $P$ in its
interior. The angle between the tangents to $D_1$ and $D_2$ at $p$ or
at $q$ (these angles are equal) is $\beta=(h-1)\alpha$;
see Figure \ref{Fig:ProvePolyg}~(a).

\begin{figure}[htbp]
\centering
\begin{tabular}{ccccc}
\input{ProvePolyg1.pspdftex}&\hspace*{2mm}&
\input{ProvePolyg2.pspdftex}&\hspace*{2mm}&
\input{ProvePolyg3.pspdftex}\\
\small (a)&&\small (b)&& \small (c)
\end{tabular}
\caption{(a) The angle between the tangents to $D_1$ and $D_2$ at
$p$ (or at $q$) is equal to $\angle w_1pw_2= \beta=(h-1)\alpha$. (b)
The disks $D$ and $D^+$ and the homothetic copy $Q_j$ of $Q$; $\ell'
\cap D =qq' \subseteq e'$. (c) $\ell'$ forms an angle of at least
$\alpha/2$ with each of the tangents to $D_1,D_2$ at $q$, and  the
edge $e'=a_1a_2 \subset D_1\cup D_2$.} 
\label{Fig:ProvePolyg}
\end{figure}
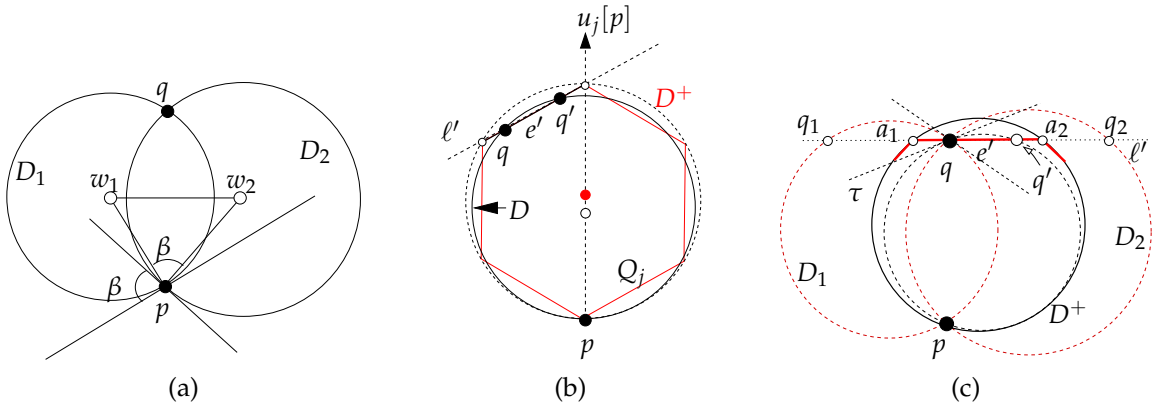

Fix an arbitrary index $1\leq j\leq h-2$, so $u_j[p]$ intersects
$e_{pq}$ and forms an angle of at least $\alpha$ with each of
${pw}_1,{pw}_2$. Let $Q_j:=Q_{pq}(u_j)$ be the $v_j$-placement 
of $Q$ at $p$ that touches $q$. To see that such a placement exists, we note that, by
the preceding remark, it suffices to show that the angle between
$\vec{pq}$ and $u_j[p]$ is at most $\pi/2-\alpha/2$; that is, to
rule out the case where $q$ lies in one of the shaded wedges in
Figure~\ref{Fig:Placement}~(c). This case is indeed impossible,
because then one of $u_{j-1}[p],u_{j+1}[p]$ would form an angle
greater than $\pi/2$ with $\vec{pq}$, contradicting the assumption
that both of these rays intersect the (Euclidean) $\bisect_{pq}$.
The claim now follows from the next lemma, which shows that $Q_j
\subset D_1\cup D_2$, which implies that $\intr Q_j \cap P
=\emptyset$ and thus $\Nbrs_j^\poly[p]=q$, as claimed.
\end{proof}

\begin{lemma}
\lemlab{polynorm-empty}
In the notation in the proof of Lemma~\ref{lemma:LongEucPoly},
$Q_j \subset D_1\cup D_2$, for $1 \le j \le h-2$.
\end{lemma}

\begin{proof}
Fix a value of $1\le j\le h-2$. Let $e'$ be the edge of $Q_j$
passing through $q$; see Figure~\ref{Fig:ProvePolyg}~(b). Let $D$ be
the disk whose center lies on $u_j[p]$ and which passes through $p$
and $q$, and let $D^+$ be the circumscribing disk of $Q_j$. Since $p
\in \partial D \cap\partial D^+$, $q\in \partial D \cap \intr D^+$,
and $D$ and $D^+$ are centered on the ray $u_j[q]$ emanating from
$p$, it follows that $D \subset D^+$. The line $\ell'$ containing
$e'$ crosses $D$ in a chord $qq'$ that is fully contained in $e'$,
as $qq'=D \cap \ell' \subseteq D^+\cap\ell' = e'$.

The angle between the tangent to $D$ at $q$, denoted by $\tau$, and
the chord $qq'$ is equal to the angle at which $p$ sees $qq'$.
This angle is smaller than the angle at which $p$ sees $e'$, which in turn
is equal to $\alpha/2$. Recall that $u_j[p]$ makes an angle of at least
$\alpha$ with each of $pw_1$ and $pw_2$, therefore $\tau$ forms an angle
of at least $\alpha$ with each of the tangents to $D_1,D_2$ at $q$.
Combining this with the fact that the angle between $\tau$ and $e'$ is
at most $\alpha/2$, we conclude that $e'$ forms an angle of at least
$\alpha/2$ with each of these tangents; see Figure~\ref{Fig:ProvePolyg}~(c).

The line $\ell'$ marks two chords $q_1q,qq_2$ within the respective
disks $D_1,D_2$. We claim that $e'$ is fully contained in their
union $q_1q_2$. Indeed, the angle $q_1pq$ is equal to the angle
between $\ell'$ and the tangent to $D_1$ at $q$, so $\angle
q_1pq\geq \alpha/2$. On the other hand, the angle at which $p$ sees
$e'$ is $\alpha/2$, which is no larger. This, and the symmetric
argument involving $D_2$, are easily seen to imply the claim.

Now consider the circumscribing disk $D^+$ of $Q_j$. Denote the
endpoints of $e'$ as $a_1$ and $a_2$, where $a_1$ lies in $q_1q$ and
$a_2$ lies in $qq_2$. Since the ray $\vec{pa}_1$ hits $\partial D^+$
before hitting $\partial D_1$, and the ray $\vec{pq}$ hits these
circles in the reverse order, it follows that the second
intersection of $\partial D_1$ and $\partial D^+$ (other than $p$)
must lie on a ray from $p$ which lies between the rays
$\vec{pa}_1,\vec{pq}$ and thus crosses $e'$. See
Figure~\ref{Fig:ProvePolyg}~(c).  Symmetrically, the second
intersection point of $\partial D_2$ and $\partial D^+$ also lies on
a ray which crosses $e'$. It follows that the arc of $\partial D^+$
delimited by these intersections and containing $p$ is fully
contained in $D_1\cup D_2$.  Hence all the vertices of $Q_j$ (which
lie on this arc) lie in $D_1\cup D_2$. This, combined with the fact,
established in the preceding paragraph, that $e'\subseteq q_1q_2$
implies that $Q_j\subset D_1\cup D_2$.
\end{proof}





Next, we use Lemma~\ref{lemma:LongEucPoly} to prove its converse. Specifically, we prove the following lemma.

\begin{lemma}\label{lem:LongPolygEuc}
Let $p,q\in P$ be a pair of points such that $\Nbrs_j^\poly[p]=q$ for
at least five consecutive indices $j\in \{0,\ldots, k-1\}$.  Then for
each of these indices, except possibly for the first two and the last two indices,
we have $\Nbrs_j[p]=q$.
\end{lemma}
\label{}

\begin{proof}
Again, assume with no loss of generality that $\Nbrs_j^\poly[p]=q$
for $0\leq j\leq h-1$, with $h\geq 5$. Suppose to the contrary that,
for some $2\leq j\leq h-3$, we have $\Nbrs_j[p]\neq q$. By
assumption, $\Nbrs^\poly_i[p]=q$, for each $0\leq i\leq h-1$, and
therefore we have $\varphi_i[p,q]\leq \varphi_i^\poly[p,q]<\infty$,
for each of these indices. In particular, we have
$\varphi_j[p,q]\leq \varphi_j^\poly[p,q]<\infty$, so there exists
$r\in P$ for which $\varphi_j[p,r]<\varphi_j[p,q]$.  Assume with no
loss of generality that $r$ lies to the left of the line from $p$ to
$q$. In this case we claim that
$\varphi_{j-1}[p,r]<\varphi_{j-1}[p,q]<\infty$ and
$\varphi_{j-2}[p,r]<\varphi_{j-2}[p,q]<\infty$.

Indeed, the boundedness of $\varphi_{j-1}[p,q]$ and $\varphi_{j-2}[p,q]$
has already been noted.  Moreover, because $r$ lies to the left of
the line from $p$ to $q$, the orientation of $\bisect_{pr}$ lies
counterclockwise to that of $\bisect_{pq}$. This, and our assumption
that $u_j[p]$ hits $\bisect_{pr}$ before hitting $\bisect_{pq}$,
implies that the point $\bisect_{pr}\cap\bisect_{pq}$ lies to the
right of the (oriented) line through $u_j[p]$; see
Figure~\ref{Fig:Converse}. Hence, any ray $\rho$ emanating from $p$
counterclockwise to $u_j[p]$ that intersects $\bisect_{pq}$ must
also hit $\bisect_{pr}$ before hitting $\bisect_{pq}$, so we have
$\varphi_{j-1}[p,r]<\varphi_{j-1}[p,q]$ and
$\varphi_{j-2}[p,r]<\varphi_{j-2}[p,q]$ (since $j \ge 2$, both
$u_{j-1}[p]$ and $u_{j-2}[p]$ intersect $\bisect_{pq}$), as claimed.
Now applying Lemma~\ref{lemma:LongEucPoly} to the point set
$\{p,q,r\}$ and to the index set $\{j-2,j-1,j\}$, we get that
$\varphi^\poly_{j-1}[p,r]<\varphi^\poly_{j-1}[p,q]$. But this
contradicts the fact that $\Nbrs_{j-1}^\poly[p]=q$. The case where
$r$ lies to the right of $\vec{pq}$ is handled in a fully symmetric
manner, using the indices $\{j,j+1,j+2\}$.
\end{proof}

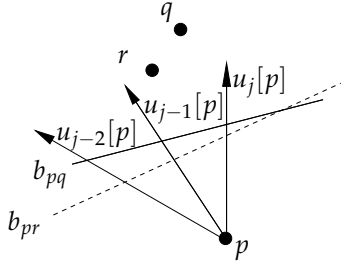
\begin{figure}[htbp]
\begin{center}
\input{Converse.pspdftex}
\caption{Proof of Lemma~\ref{lem:LongPolygEuc}.}
\label{Fig:Converse}
\end{center}
\end{figure}

Combining Lemmas~\ref{lemma:LongEucPoly} and~\ref{lem:LongPolygEuc}, we obtain the following stronger version
of Theorem~\ref{thm:Qnorm}.
\begin{theorem}
Let $P$ be a set of $n$ points in $\reals^2$, $\alpha>0$ a parameter,
and $Q$ a regular $k$-gon with $k \ge 2\pi/\alpha$.
Then the following properties hold.
\begin{itemize}
\item[(i)] Every $4\alpha$-stable edge in $\DT(P)$
is an $\alpha$-stable edge in $\DT^\poly(P)$.
\item[(ii)] Every $6\alpha$-stable edge in $\DT^\poly(P)$
is also an $\alpha$-stable edge in $\DT(P)$.
\end{itemize}
\end{theorem}
\begin{proof}
Let $pq$ be a $4\alpha$-stable edge of $\DT(P)$.
Then the corresponding edge $e_{pq}$ in $\VD(P)$ stabs at least four rays 
$u_j[p]$ emanating from $p$,
and, by Lemma~\ref{lemma:LongEucPoly}, $N_j^\poly[p]=q$ for at least
two of these values of $j$. Therefore, $p$ sees the edge $e_{pq}^\poly$
in $\VD^\poly(P)$ at an angle at least $\alpha$. Similarly $q$ sees the edge $e_{pq}^\poly$ at
an angle at least $\alpha$.
Conversely, if $pq$ is $6\alpha$-stable in $\DT^\poly(P)$ then $e^\poly_{pq}$
meets at least six rays $u_j[p]$, and then Lemma~\ref{lem:LongPolygEuc} is easily seen
to imply that $p$ (and, symmetrically, $q$ too) sees $e_{pq}$ at an angle at least $\alpha$. 
\end{proof}

The next lemma gives a slightly different characterization of stable edges, which is more 
algorithmic and which will be useful in maintaining a SDG under a constant-degree algebraic
motion of the points of $P$.
\begin{lemma}
\lemlab{poly-stable}
Let $\G$ be the subgraph of $\DT^\poly(P)$ composed of the edges whose
dual $Q$-Voronoi edges contain at least eleven breakpoints. Then
$\G$ is an $(8\alpha,\alpha)$-$\SDG$ of $P$ (in the Euclidean norm).
\end{lemma}

\begin{proof}
Let $p,q\in P$ be two points. If $(p,q)$ is an $8\alpha$-stable edge
in $\DT(P)$ then the dual Voronoi edge $e_{pq}$  stabs at least eight rays
$u_j[p]$ emanating from $p$, and at least eight rays $u_j[q]$ emanating from $q$.
Lemma~\ref{lemma:LongEucPoly} implies that  $\VD^\poly(P)$ contains the edge
$e_{pq}^\poly$ with at least six breakpoints corresponding to corner
placements of $Q$ at $p$ that touch $q$, and at least six breakpoints
corresponding to corner placements of $Q$ at $q$ that touch $p$.
Therefore, $e_{pq}^\poly$ contains at least twelve breakpoints,
so $(p,q)\in \G$.

Conversely, suppose $p,q\in P$ define an edge $e_{pq}^\poly$ in $\VD^\poly(P)$
with at least eleven breakpoints. By the interleaving property of 
breakpoints, stated in \lemref{interleave},
we may assume, without loss of generality,
that at least six of these breakpoints correspond to empty corner placements
of $Q$ at $p$ that touch $q$.  Lemma~\ref{lem:LongPolygEuc} implies that
$\VD(P)$ contains the edge $e_{pq}$, and that this edge is hit by at least
two consecutive rays $u_j[p]$. But then the $(p,q)$ is $\alpha$-stable in $\DT(P)$.
\end{proof}

In a companion paper~\cite{polydel}, we describe a kinetic data structure (KDS) 
for maintaining $\DT^\poly(P)$. As shown in that paper, it
can also keep track of the number of breakpoints for each edge of $\DT^\poly(P)$. 
If each point of $P$ moves along an algebraic trajectory of bounded
degree, then the KDS processes $O(k^4n\lambda_r(n))$ events, where
$r$ is a constant depending on the complexity of motion of $P$. A change in 
the number of breakpoints in a Voronoi edge is an event that the KDS can detect and process.
As discussed in detail in~\cite{polydel}, many events, so-called \emph{singular} events,
that occur when an edge of $\DT^\poly(P)$ becomes parallel to an edge of $Q$, can occur 
simultaneously. Nevertheless, each of the events can be processed in $O(\log n)$ time,
and their overall number is within the bound cited above.
We maintain the subgraph $\G$ of $\DT^\poly(P)$, consisting of the 
edges of $\DT^\poly(P)$ that have at least \emph{eleven} breakpoints, which,
by \lemref{poly-stable}, is an $(8\alpha,\alpha)$-Euclidean SDG.
Putting everything together, we obtain a KDS that maintains an $(8\alpha,\alpha)$-SDG of $P$. 
It uses linear storage, it processes 
$O\left(\frac{1}{\alpha^4}n\lambda_r(n)\right)$ events, where $r$
is a constant that depends on the degree of the motion of the points
of $P$, and it updates the SDG at each event in $O(\log n)$ time.
This proves Theorem~\ref{Thm:MaintainSDGPolyg}.

\medskip 
\noindent\textbf{Remarks.} 
(i) We remark that $\DT^\poly(P)$ can undergo $\Omega(n)$ 
changes at a time instance when $\Omega(n)$ singular events occur 
simultaneously, say, when $pq$ becomes parallel to an edge of $Q$,
but all these changes occur at the edges incident to $p$ or $q$ in $\DT(P)$. 
However, only $O(1/\alpha)$ edges among them can have at least eleven
breakpoints, before or after the event.  Hence, $O(1/\alpha)$ edges can simultaneously
enter or leave the $(8\alpha,\alpha)$-SDG $\G$ of Theorem \ref{Thm:MaintainSDGPolyg}. 

\smallskip
\noindent
(ii) Note that there is a slight discrepancy between the value of $k$ that
we use in this section ($k\ge 2\pi/\alpha$), and the value needed to ensure that 
the regular $k$-gon $Q_k$ is $\alpha$-close to the Euclidean disk, which is
$k\ge \pi/\alpha$. This is made for the convenience of presentation.


\smallskip
\noindent
(iii) An interesting open problem is whether the dependence on $\alpha$ can be
improved in the above KDS. We have developed an
alternative scheme for maintaining stable (Euclidean) Delaunay
graphs, which  is reminiscent of the kinetic schemes used by
Agarwal~\etal~\cite{KineticNeighbors} for maintaining closest pairs
and nearest neighbors. It extracts a nearly linear number of pairs of points of $P$
that are candidates for being stable Delaunay edges and then sifts the stable
edges from these candidate pairs using the so-called 
\emph{kinetic tournaments}~\cite{bgh-dsmd-99}.  Although the overall 
structure is not complicated, the analysis is rather technical and lengthy, so we
omit this KDS from this version of the paper; it can be
found in the arXiv version~\cite{StableFull}. In summary, the resulting KDS is of size 
$O((n/\alpha^2)\log n)$, it processes a total of
$O((n/\alpha)^2\beta_r(n/\alpha)\log^2n \log(\log(n)/\alpha))$ events,
and it takes a total of 
$O((n/\alpha)^2\log^2 n(\log^2 n + \beta_r(n/\alpha)\log^2n \log^2(\log(n)/\alpha)))$ 
time to process them; here $\beta_r(n)=\lambda_r(n)/n$ is an extremely slowly growing function for any fixed $r$. 
The worst-case time of processing an event is $O(\log^4(n)/\alpha))$. Another advantage of this data
structure is that, unlike the above KDS, it is \emph{local} in the terminology of
\cite{bgh-dsmd-99}.  Specifically,
each point of $P$ is stored, at any given time, at only $O(\log^2 (n)/\alpha^2)$ places in
the KDS. Therefore the KDS can efficiently accommodate an update in the trajectory of a point.

\section{Stability under Nearly Euclidean Distance Functions}
\label{sec:Qnorm}

In this section we prove Theorem~\ref{thm:Qnorm} for an arbitrary convex 
distance function that is $\alpha$-close to the Euclidean norm (see the
Introduction for the definition).

Let $Q$ be a compact, convex set that contains the origin in its interior, 
and let $d_Q$ denote the distance function induced by $Q$. Assume that $Q$ is 
$\alpha$-close to the Euclidean norm.

For sake of brevity, we carry out the proof assuming that $Q$ is strictly convex, i.e.,
the relative interior of the chord connecting two point
$x$ and $y$ on $\bd Q$ is strictly contained in the interior of $Q$ 
(there are no straight segments on the boundary of $Q$).
The proof also holds verbatim when $Q$ is not strictly convex, provided that
no pair of points $p,q\in P$ is such that $\vec{pq}$ is parallel to a 
straight portion of $\bd Q$.

We assume that $P$ is in general position with respect to $Q$, in the sense that
no three points of $P$ lie on a line, and no four points of $P$ are
\emph{$Q$-cocircular}, i.e., no four points of $P$ lie on the
boundary of a common homothetic copy of $Q$.

Recall that for a direction $u$ and for a point $p \in P$, $Q_p(u)$ denotes a ``generic''
homothetic copy of $Q$ that touches $p$ and is centered at some
point on $u[p]$. See Figure~\ref{Fig:Qcopies}~(a). As mentioned in Section~\ref{sec:prelim}, all
homothetic copies of $Q_p(u)$ touch $p$ at (points corresponding to) the same point 
$\zeta\in\bd Q$ and therefore share the same tangent $\ell_p(u)$ at $p$.
If $Q$ is not smooth at $\zeta$, there is a range of possible orientations of such tangents.
In this situation we let $\ell_p(u)$ denote an arbitrary tangent of this kind.
In addition, if $Q$ is smooth at $\zeta$ then $Q_{pq}(u)$ exists for any point $q$ that
lies in the same side of $\ell_p(u)$ as $Q_p(u)$. Otherwise, this has to hold for every
possible tangent $\ell_p(u)$ at $\zeta$, which is equivalent to requiring that $q$
lies in the wedge formed by the intersection of the two halfplanes bounded by the two
extreme tangents at $\zeta$ and containing $Q_p(u)$.

\smallskip
\noindent
{\bf Remarks.}
(1) An important observation is that, when $q$ satisfies these conditions, $Q_{pq}(u)$ is unique,
unless all the three following conditions hold: (i) $Q$ is not strictly convex, (ii) $\vec{pq}$ 
is parallel to straight portion $e$ of $\bd Q$, and (iii) $u$ is a direction connecting some 
point on $e$ to the center $c$ of $Q$. We leave the straightforward proof of this property to the reader.
The proof of the theorem exploits the uniqueness of $Q_{pq}(u)$, and breaks down when it is not unique.
In fact, this is the only way in which the assumptions concerning strict convexity are used in the proof.

\smallskip
\noindent (2) With some care, our analysis applies also if (the directions of) some pairs $pq$ are parallel to straight portions of $\partial Q$, in which case $Q_{pq}(u)$ is not uniquely defined for certain directions $u$.
This extension requires the more elaborate notion of $\alpha$-stability, which ignores the possible two-dimensional portions of $e_{pq}^Q$; see Section \ref{sec:prelim} for more details. Informally, this allows
us to avoid the ``problematic" directions $u$ in which $Q_{pq}(u)$ is not unique. (The latter happens exactly when $u[p]$ hits $b_{pq}^Q$ within one of its two-dimensional portions.) We note, though, that the loss in the amount of stability caused by ignoring a two-dimensional portion of $e_{pq}^Q$ is
at most $2\alpha$. which is an upper bound on the angular span of directions that connect a straight portion of $\bd Q$ to its center. This latter property holds since $Q$ is $\alpha$-close to the Euclidean disk; see below for more details.

\medskip
The proof of Theorem~\ref{thm:Qnorm} relies on the following 
three simple geometric properties. Recall that the $\alpha$-closeness of $Q$ to the Euclidean norm
means that $D_I = (\cos\alpha)D_O\subseteq Q\subseteq D_O$.

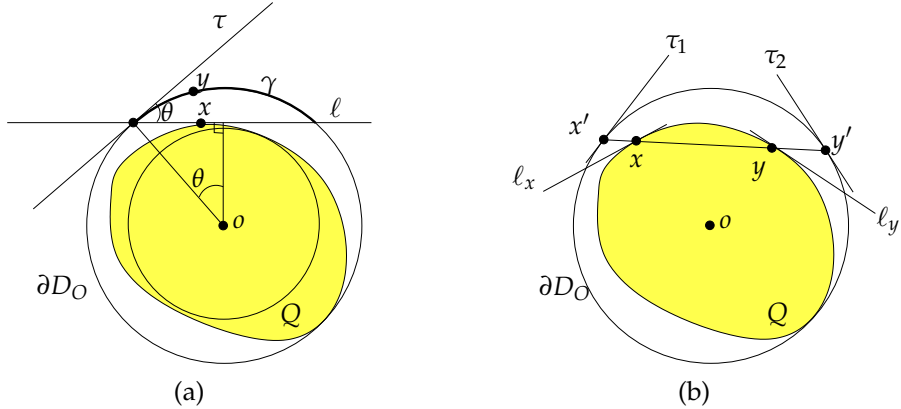
\begin{figure}[hbt]
\begin{center}
\begin{tabular}{ccc}
\input{norm.pspdftex}&
\hspace*{1cm}&
\input{ClaimNorm.pspdftex}\\
\small (a)&&\small (b)
\end{tabular}
\caption{Illustrations for: (a) Claim~\ref{Q1}; (b) Claim~\ref{Q3}.}
\label{fig:norm1}
\end{center}
\end{figure}

\begin{claim} \label{Q1}
Let $x$ be a point on $\bd Q$, and let
$\ell$ be a supporting line to $Q$ at $x$.
Let $y$ be the point on $\bd D_O$
closest to $x$ ($x$ and $y$ lie on the same radius from the center $o$), and
let $\gamma$ be the arc of $\bd D_O$ that contains $y$ and is bounded by
the intersection points of $\ell$ with  $\bd D_O$.
Then the angle between $\ell$ and  the tangent, $\tau$, to $D_O$ at any point
along $\gamma$,
is at most $\alpha$.
\end{claim}
\begin{proof}
Denote this angle by $\theta$. Clearly $\theta$ is maximized when
$\tau$ is tangent to $D_O$ at an intersection of $\ell$ and $\bd
D_O$; see Figure~\ref{fig:norm1}~(a). For this value of $\theta$, it is easy to verify that
the distance from $o$ to $\ell$ is $\cos\theta$. But this distance
has to be at least $\cos\alpha$, because $\ell$ fully contains $D_I
= (\cos\alpha) D_O$ on one side. Hence $\cos\theta \geq \cos\alpha$,
and thus $\theta \leq \alpha$, as claimed.
\end{proof}

\noindent{\bf Remark.}
An easy consequence of this claim is that the angle in which the center of $Q$ sees any straight portion of $\bd Q$ (when $Q$ is not strictly convex) is at most $2\alpha$.

\begin{claim} \label{Q3}
Let $x$ and $y$ be two points on $\bd Q$, and let $\ell_x$ and
$\ell_y$ be supporting lines of $Q$ at $x$ and $y$, respectively.
Then the difference between the (acute) angles that $\ell_x$ and
$\ell_y$ form with $xy$ is at most $2\alpha$.
\end{claim}
\begin{proof} Denote the two angles in the claim by $\theta_x$ and
$\theta_y$, respectively. Continue the segment $xy$ beyond $x$ and beyond $y$
until it intersects $D_O$ at $x'$ and $y'$, respectively. 
Let $\tau_1$ and $\tau_2$ denote the respective
tangents to $D_O$ at $x'$ and at $y'$. See Figure
\ref{fig:norm1}~(b). Clearly, the respective angles $\theta_1$,
$\theta_2$ between the chord $x'y'$ of $D_O$ and $\tau_1$, $\tau_2$
are equal.  By Claim~\ref{Q1} (applied once to $\tau_1$ and $\ell_x$ and once
to $\tau_2$ and $\ell_y$) we get that $|\theta_1-\theta_x|\le\alpha$ and
$|\theta_2-\theta_y|\le\alpha$, and the claim follows.
\end{proof}



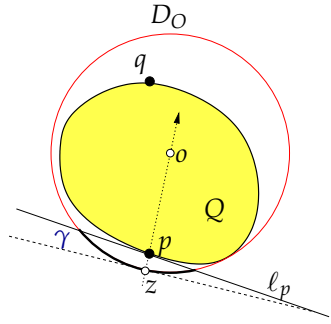
\begin{figure}[htb]
\begin{center}
\input{CloseTangents.pspdftex}
\caption{Proof of Claim \ref{Claim:TangentsDQ}.}
\label{fig:norm3}
\end{center}
\end{figure}

\begin{claim} \label{Claim:TangentsDQ}
For a point $p\in\bd Q$, any tangent $\ell_{p}$
to $Q$ at $p$ forms an angle at most $\alpha$ with any line
orthogonal to $\vec{op}$.
\end{claim}

\begin{proof}
See Figure~\ref{fig:norm3}.
Consider the chord $\xi = \ell_p\cap D_{O}$, and let $\gamma$
denote the arc of $\bd D_O$ determined by $\xi$ and containing
the intersection point $z$ of $\bd D_{O}$ and the ray $\vec{op}$.
By Claim~\ref{Q1}, the angle between $\ell_{p}$ and the tangent
to $D_{O}$ at $z$ is at most $\alpha$. Since this tangent
is orthogonal to $\vec{op}$, the claim follows.
\end{proof}

\noindent{\bf Remark.}
Clearly, Claims \ref{Q1}--\ref{Claim:TangentsDQ} continue to hold
for any homothetic copy of $Q$, with a corresponding translation
and scaling of $D_O$ and $D_I$.

\medskip

Let $Q_{pq}^-(u)$ (resp., $Q_{pq}^+(u)$) denote the portion of
$Q_{pq}(u)$ that lies to the left (resp., right) of the directed
line from $p$ to $q$. 
Let $D_{pq}(u)$ denote the disk that touches $p$ and $q$, and whose center lies on $u[p]$.

We next establish the following lemma, whose setup is illustrated in
Figure \ref{fig:qcdv}~(a). It provides the main geometric
ingredient for the proof of Theorem~\ref{thm:Qnorm}.

\begin{lemma}\label{Lemma:ContainQ}
(i) Let $u\in \sphere^1$ be a direction such that both
$Q_{pq}(u)$ and $D_{pq}(u+5\alpha)$ are defined.
Then the region $Q_{pq}^+(u)$ is fully contained in the disk $D_{pq}(u+5\alpha)$.

\smallskip
\noindent(ii) Let $u\in \sphere^1$ be a direction such that both
$Q_{pq}(u)$ and $D_{pq}(u-5\alpha)$ are defined. Then the region
$Q_{pq}^-(u)$ is fully contained in the disk $D_{pq}(u-5\alpha)$.
\end{lemma}
\begin{proof}
It suffices to establish Part (i) of the lemma; the proof of the
other part is fully symmetric.

Refer to Figure \ref{fig:qcdv}~(b). Let $\ell_p=\ell_p(u)$ be any
supporting line of $Q_{pq}(u)$ at $p$, as defined above, and let
$\tau_p=\tau_p(u)$ be the line through $p$ that is orthogonal to $u$
(which is also the tangent to $D_{pq}(u)$). By Claim~\ref{Claim:TangentsDQ}, 
the angle between $\ell_p$ and $\tau_p$ is
at most $\alpha$. We next consider the tangent $\tau_p^+$ to
$D_{pq}(u+5\alpha)$ at $p$. Since the angle between $\tau_p$ and
$\tau_p^+$ is $5\alpha$, it follows that the angle between $\ell_p$
and $\tau_p^+$ is at least $4\alpha$.
The preceding arguments imply that, when oriented into the right
side of $\vec{pq}$, $\ell_p$ lies between $\vec{pq}$ and $\tau_p^+$,
and the angle between $\ell_p$ and $\tau_p^+$ is at least
$4\alpha$.

This implies that, locally near $p$, $\bd Q^+_{pq}(u)$ penetrates
into $D_{pq}(u+5\alpha)$. This also holds at $q$. To establish the
claim for $q$ (which is not symmetric to the claim for $p$, because the
center $c$ of $Q_{pq}(u)$ lies on the ray $u[p]$ emanating from $p$,
and there is no control over the orientation of the corresponding
ray $\vec{qc}$ emanating from $q$), we note that, by Claim~\ref{Q3},
the angles between $pq$ and any pair of tangents $\ell_p$, $\ell_q$ to
$Q_{pq}(u)$ at $p$, $q$, respectively, differ by at most $2\alpha$,
whereas the angles between $pq$ and the two tangents $\tau_p^+$,
$\tau_q^+$ to $D_{pq}(u+5\alpha)$ at $p$, $q$, respectively, are
equal. This, and the fact that the angle between $\tau_p^+$ and
$\ell_p$ is at least $4\alpha$, imply that, when oriented into the
right side of $\vec{pq}$, $\ell_q$ lies between $\vec{qp}$ and
$\tau_q^+$, which thus implies the latter claim. Note also that the
argument just given ensures that the angle between $\tau_q^+$ and
$\ell_q$ is at least $2\alpha$.

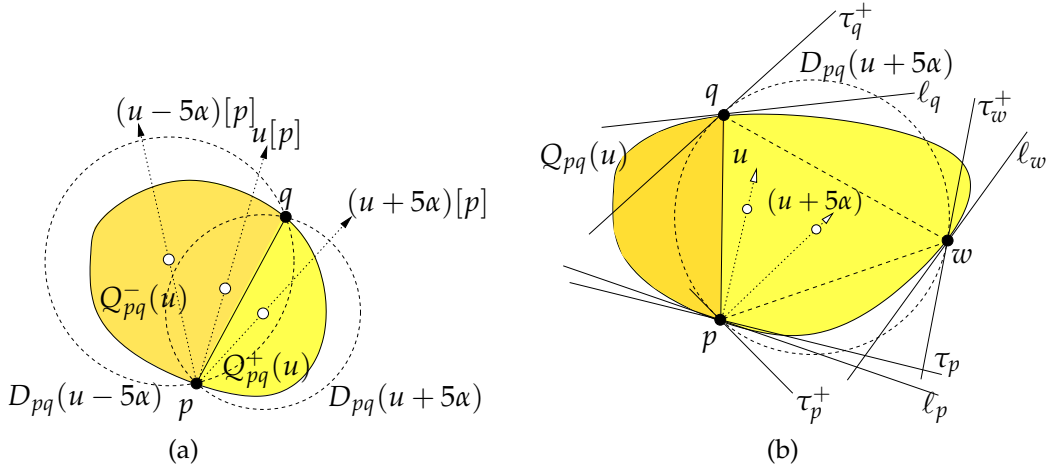
\begin{figure}[hbt]
\centering
\begin{tabular}{ccc}
\input{TwoParts.pspdftex}&\hspace*{1.5cm}&\input{qcdv1.pspdftex}\\
\small (a)&&\small (b)
\end{tabular}
\caption{(a) The setup of Lemma \ref{Lemma:ContainQ};
(b) Proof of Lemma \ref{Lemma:ContainQ}~(i): $\bd Q^+_{pq}(u)$
cannot cross $\bd D_{pq}(u+5\alpha)$ at any third point $w$.}
\label{fig:qcdv}
\end{figure}

It therefore suffices to show that $\bd Q_{pq}^+(u)$ does not cross
$\bd D_{pq}(u+5\alpha)$ at any third point (other than $p$ and $q$).
Suppose to the contrary that there exists such a third point $w$,
and consider the tangents $\ell_w$ to $Q_{pq}(u)$ at $w$, and
$\tau_w^+$ to $D_{pq}(u+5\alpha)$ at $w$. Consider the two points
$p$ and $w$, and apply to them an argument similar to the one used
above for $p$ and $q$. Specifically, we use the facts that (i) the
angles between $pw$ and $\tau_p^+$, $\tau_w^+$ are equal, (ii) the
angles between $pw$ and $\ell_p$, $\ell_w$, for any tangent $\ell_w$ 
to $Q_{pq}(u)$ at $w$, differ by at most
$2\alpha$, and (iii) the angle between $\ell_p$ and $\tau_p^+$ is at
least $4\alpha$, to conclude that, when oriented into the left side
of $\vec{wp}$, $\ell_w$ lies strictly between $\vec{wp}$ and
$\tau_w^+$. See Figure~\ref{fig:qcdv}. Similarly, applying the
preceding argument to $q$ and $w$, we now use the facts that (i) the
angles between $wq$ and $\tau_q^+$, $\tau_w^+$ are equal, (ii) the
angles between $qw$ and $\ell_q$, $\ell_w$ differ by at most
$2\alpha$, and (iii) the angle between $\tau_q^+$ and $\ell_q$ is at
least $2\alpha$, to conclude that, when oriented into the right side
of $\vec{wq}$, $\ell_w$ lies between $\vec{wq}$ and $\tau_w^+$ or
coincides with $\tau_w^+$. This impossible configuration shows that
$w$ cannot exist, and consequently that $Q_{pq}^+(u)\subset
D_{pq}(u+5\alpha)$.
\end{proof}

\paragraph{Proof of Theorem \ref{thm:Qnorm} -- Part (i).}
Let $pq$ be an $11\alpha$-stable edge in the Euclidean Delaunay
triangulation $\DT(P)$.  That is, the Euclidean Voronoi edge
$e_{pq}$ is hit by two rays $u^-[p],u^+[p]$ which form an angle of
at least $11\alpha$ between them (where $u^+$ is assumed to lie
counterclockwise to $u^-$). Clearly, $e_{pq}$ is also hit by any ray $u[p]$
whose direction $u$ belongs to the interval $(u^-,u^+)\subset
\sphere^1$.  Let $u[p]$ be such a ray whose direction $u$ belongs to
the interval $(u^-+5\alpha,u^+-5\alpha)$ (of span at least
$\alpha$).  That is, $u[p]$ hits $e_{pq}$ ``somewhere in the
middle'', so all the three disks $D_{pq}(u-5\alpha),D_{pq}(u)$ and
$D_{pq}(u+5\alpha)$ are defined and contain no points of $P$ in
their respective interiors. (Actually, $D_{pq}(u)$ is contained in
$D_{pq}(u-5\alpha)\cup D_{pq}(u+5\alpha)$, as is easily checked.)

We next consider the $Q$-Voronoi diagram $\VD^Q(P)$. We claim that
the corresponding edge $e^Q_{pq}$ exists and is also hit by $u[p]$.
Since this holds for every $u\in (u^-+5\alpha,u^+-5\alpha)$,
it follows that $(p,q)$ is $\alpha$-stable in $\VD^Q(P)$.

To establish this claim, we prove the following two properties.
\begin{itemize}
\item[(i)] the homothetic copy $Q_{pq}(u)$ exists, and 
\item[(ii)] it contains no points of $P$ in its interior.
\end{itemize}
\medskip

\noindent
\textbf{\textit{Proof of (i):}}\hspace*{3mm}
Assume to the contrary that the copy $Q_{pq}(u)$
is undefined. Consider the respective tangents $\ell_p(u)$ and
$\tau_p(u)$ to $Q_{p}(u)$ and $D_{pq}(u)$ at $p$, where $\ell_p(u)$ is
any tangent to $Q_p(u)$ at $p$ that separates $q$ from $Q_p(u)$;
such a tangent exists if and only if $Q_{pq}(u)$ is undefined. 
(As noted before, $\ell_p(u)$ does not depend on the location of 
the center $c$ of $Q$ on $u[p]$.) By Claim \ref{Claim:TangentsDQ}, 
the angle between $\ell_p(u)$ and $\tau_p(u)$ is at most $\alpha$. 
Since $Q_{pq}(u)$ is undefined, the choice of $\ell_p(u)$ guarantees 
that $q$ lies inside the open halfplane $h_p(u)$ bounded by
$\ell_{p}(u)$ and disjoint from $u[p]$.

Let $\mu_p(u+5\alpha)$ (resp., $\mu_p(u-5\alpha)$) denote the open
halfplane bounded by $\tau_p(u+5\alpha)$ (resp.,
$\tau_p(u-5\alpha)$) and disjoint from the disk $D_{pq}(u+5\alpha)$
(resp., $D_{pq}(u-5\alpha)$). Since each of the lines
$\tau_p(u+5\alpha)$, $\tau_p(u-5\alpha)$ makes an angle of at least
$5\alpha$ with $\tau_p(u)$, the halfplane $h_p(u)$ supported by
$\ell_{p}(u)$ is contained in the union $\mu_p(u+5\alpha)\cup
\mu_p(u-5\alpha)$. Since $q$ is contained in $h_p(u)$, at least one
of these latter halfplanes, say
$\mu_p(u+5\alpha)$, must contain $q$.
However, if $q\in \mu_p(u+5\alpha)$, the corresponding copy
$D_{pq}(u+5\alpha)$ is undefined, a contradiction that establishes (i).
\medskip

\noindent
\textbf{\textit{Proof of (ii):}}\hspace*{3mm}
Since both $Q_{pq}(u)$ and $D_{pq}(u+5\alpha)$ are defined, 
Lemma \ref{Lemma:ContainQ}(i) implies that 
$Q_{pq}^+(u) \subset D_{pq}(u+5\alpha)$. Moreover  the interior 
of $D_{pq}(u+5\alpha)$ is $P$-empty, so the interior of $Q_{pq}^+(u)$ is also $P$-empty.
A symmetric argument (using Lemma \ref{Lemma:ContainQ}(ii)) implies that
the interior of $Q_{pq}^-(u)$ is also $P$-empty. 

This completes the proof of part (i) of Theorem~\ref{thm:Qnorm}.

\paragraph{Proof of Theorem \ref{thm:Qnorm} -- Part (ii).}
We fix a direction $u\in \sphere^1$ for which all the three copies
$Q_{pq}(u)$, $Q_{pq}(u-5\alpha)$, and $Q_{pq}(u+5\alpha)$ are
defined and have $P$-empty interiors. Again, 
$Q_{pq}(u) \subset Q_{pq}(u-5\alpha) \cup Q_{pq}(u+5\alpha)$.
Since $pq$ is $11\alpha$-stable under $d_Q$, there is an arc on 
$\sphere^1$ of length at least $\alpha$, so that every $u$ in this arc has this
property. We need to show that, for every such $u$,
\begin{itemize}
\item[(i)] the copy $D_{pq}(u)$ is defined, and 
\item[(ii)] its interior is $P$-empty. 
\end{itemize}
Similar to the proof of Part~(i), this would imply that the ray $u[p]$ hits 
the edge $e_{pq}$ of $\VD(P)$ for every $u$ in an arc of length $\alpha$, so
$pq$ is $\alpha$-stable in $\DT(P)$, as claimed.
\medskip

\noindent
\textbf{\textit{Proof of (i):}}\hspace{3mm}
Assume to the contrary that $D_{pq}(u)$ is undefined, so
the angle between the vectors $\vec{pq}$ and $u$ is at least
$\pi/2$. Let $\ell_p(u-5\alpha)$, $\ell_p(u)$, and
$\ell_p(u+5\alpha)$ be any triple of respective tangents to
$Q_{pq}(u-5\alpha)$, $Q_{pq}(u)$, and $Q_{pq}(u+5\alpha)$ at $p$. 
Let $h_p(u-5\alpha)$ (resp., $h_p(u+5\alpha)$) be the open halfplane
supported by $\ell_p(u-5\alpha)$ (resp., $\ell_p(u+5\alpha)$) and
disjoint from $Q_{pq}(u-5\alpha)$ (resp., $Q_{pq}(u+5\alpha)$).
Claim \ref{Claim:TangentsDQ} implies that each of
the lines $\ell_p(u-5\alpha)$, $\ell_p(u+5\alpha)$ makes with
$\tau_{p}(u)$, the line orthogonal to $u[p]$ at $p$, an angle of at
least $4\alpha$ (and at most $6\alpha$). Indeed, the claim implies
that the angle between $\ell_p(u-5\alpha)$ and the line
$\tau_p(u-5\alpha)$, which is orthogonal to $(u-5\alpha)[p]$ at $p$, is at
most $\alpha$. Since the angle between $\tau_p(u-5\alpha)$ and
$\tau_p(u)$ is $5\alpha$, the claim for $\ell_p(u-5\alpha)$ follows.
A symmetric argument establishes the claim for $\ell_p(u+5\alpha)$.
Therefore, the halfplane $\mu_p(u)$, supported by $\tau_p(u)$ and
containing $q$, is covered by the union of $h_p(u-5\alpha)$ and
$h_p(u+5\alpha)$. We conclude that at least one of these latter
halfplanes must contain $q$. However, this contradicts the
assumption that both copies $Q_{pq}(u-5\alpha)$, $Q_{pq}(u+5\alpha)$
are defined, and (i) follows.
\medskip

\noindent
\textbf{\textit{Proof of (ii):}}\hspace{3mm}
Assume to the contrary that $D_{pq}(u)$, whose
existence has just been established, contains some point $r$ of $P$
in its interior. That is, the ray $u[p]$ hits $b_{pr}$ before
$b_{pq}$. In this case $D_{pr}(u)$ also exists. With no
loss of generality, we assume that $r$ lies to the left of the
oriented line from $p$ to $q$.


We claim that the homothetic copy $Q_{pr}(u-5\alpha)$ exists and
contains $q$.  Indeed, since $Q_{pq}(u-5\alpha)$ exists and is
$P$-empty, it follows that $(u-5\alpha)[p]$ either hits $b^Q_{pr}$
after $b^Q_{pq}$ (in which case the claim obviously holds) or misses
$b_{pr}^Q$ altogether. Suppose that $(u-5\alpha)[p]$ misses
$b^Q_{pr}$. As argued earlier, this means that there exists a tangent
$\ell_p(u-5\alpha)$ to $Q_p(u-5\alpha)$ at $p$, such that $r$ lies in 
the open halfplane $h_p(u-5\alpha)$ supported by $\ell_p(u-5\alpha)$ 
and disjoint from $Q_{p}(u-5\alpha)$. 

By applying Claim \ref{Claim:TangentsDQ} as before we get that the
tangent $\tau_p(u)$ to $D_{pq}(u)$ (at $p$) to the left of
$\vec{pq}$ is between $\vec{pq}$ and $\ell_p(u-5\alpha)$ and makes
with $\ell_p(u-5\alpha)$ an angle of at least $4\alpha$. It follows
that the wedge formed by the intersection of $\ell_p(u-5\alpha)$
 and the halfplane to the left of $\vec{pq}$  is fully contained in the halfplane
$\mu_{p}(u)$ that is supported by $\tau_p(u)$ and disjoint from
$D_{pq}(u)$; see Figure \ref{fig:ReverseQ}~(a). But then
$D_{pr}(u)$ is undefined, a contradiction that
implies the existence of $Q_{pr}(u-5\alpha)$.

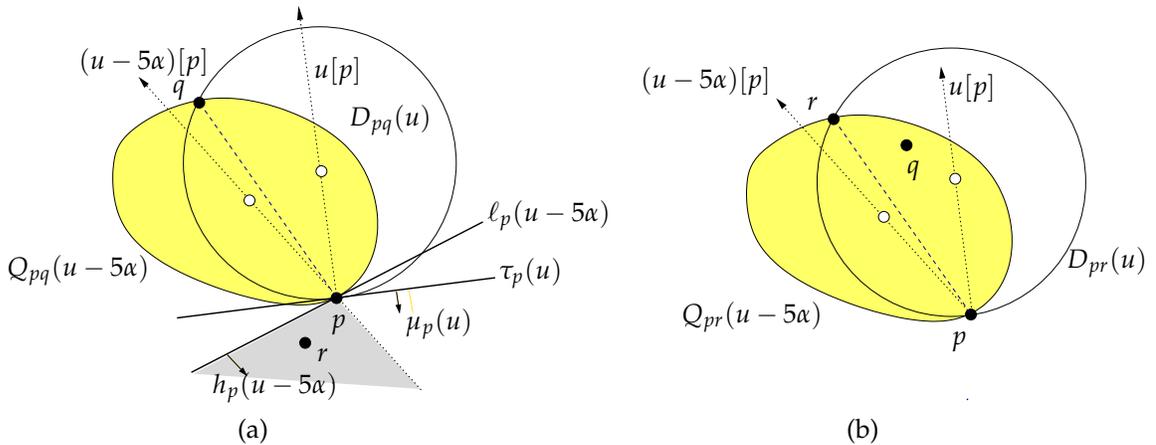
\begin{figure}[htb]
\centering
\begin{tabular}{ccc}
\input{CopyExists.pspdftex}&\hspace*{1.0cm}&\input{SecondDirection.pspdftex}\\
\small (a)&&\small (b)
\end{tabular}
\caption{Proof of Theorem \ref{thm:Qnorm}(ii): (a) Arguing (by
contradiction) that $Q_{pr}(u-5\alpha)$ exists; (b) the copy
$Q_{pr}(u-5\alpha)$ is defined and contains $q$. Hence, the disk
$D_{pr}(u)$, which covers $Q^+_{pr}(u-5\alpha)$, must contain $q$ as
well.} 
\label{fig:ReverseQ}
\end{figure}


We can now assume that $Q_{pr}(u-5\alpha)$ is defined and contains
$q$. More precisely, $q$ lies in the portion $Q^+_{pr}(u-5\alpha)$
of $Q_{pr}(u-5\alpha)$, since $q$ lies to the right of the oriented
line from $p$ to $r$.  However, Lemma \ref{Lemma:ContainQ}(i),
applied to $u-5\alpha$, implies that $Q^+_{pr}(u-5\alpha)$ is
contained in $D_{pr}(u)$, so $D_{pr}(u)$ also contains $q$; see
Figure \ref{fig:ReverseQ}~(b). It is however impossible for
$D_{pr}(u)$ to contain $q$ and for $D_{pq}(u)$ to contain $r$. This
contradiction concludes the proof of part~(ii) of Theorem~\ref{thm:Qnorm}. $\Box$


\section{Properties of stabe Delaunay graphs} 
\label{sec:SDGProperties}

We establish a few useful properties of stable Delaunay graphs in this section.

\paragraph{Near cocircularities do not show up in an SDG.}
Consider a critical event during the kinetic maintenance of $\DT(P)$,
in which four points $a,b,c,d$ become cocircular,
in this order, along their circumcircle, with this circle being $P$-empty.
Just before the critical event, the Delaunay triangulation contained
two triangles formed by this quadruple, say, $abc$, $acd$.
The Voronoi edge $e_{ac}$ then shrinks
to a point (namely, to the circumcenter of $abcd$ at the critical event),
and, after the critical cocircularity, is replaced by the Voronoi edge
$e_{bd}$, which expands from the circumcenter as time progresses.

Our algorithm will detect the possibility of such an event before the
criticality occurs, when $ac$ ceases to be $\alpha$-stable (or even before
this happens). It will then remove this edge from the stable subgraph,
so the actual cocircularity will not be recorded. The new edge $bd$
will then be detected by the algorithm only when it becomes at least $\alpha$-stable
(if this happens at all), and will then enter the stable Delaunay graph.
In short, critical cocircularities do not arise {\em at all} in our scheme.

As noted in the introduction, a Delaunay edge $ab$ (interior to the
hull) transitions from being $\alpha$-stable to not being stable, or vice-versa,
when the sum of the opposite angles in its two adjacent Delaunay
triangles is $\pi-\alpha$ (see Figure~\ref{Fig:LongDelaunay}). This
shows that changes in the stable Delaunay graph occur when the
``cocircularity defect'' of a nearly cocircular quadruple (i.e., the
difference between $\pi$ and the sum of opposite angles in the
quadrilateral spanned by the quadruple) is between $\alpha$ and
$8\alpha$.
Note that a degenerate case of cocircularity is a collinearity on
the convex hull. Such collinearities also do not show up in the
stable Delaunay graph. A hull collinearity between three nodes $a,
b, c$ is detected before it happens, when (or before) the
corresponding Delaunay edge is no longer $\alpha$-stable, in which case the
angle $\angle acb$, where $c$ is the middle point of the
(near-)collinearity becomes $\pi-\alpha$ (see
Figure~\ref{collinearity}(a)). Therefore a hull edge is removed from
the $\SDG$ if the Delaunay triangle is almost flat. The edge (or any
new edge about to replace it) re-appears in the $\SDG$ when it 
becomes $c\alpha$-stable, for some $1 \le c \le 8$.

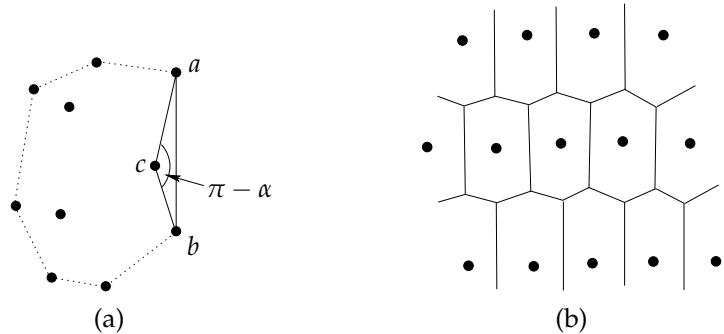
\begin{figure}[htbp]
\begin{center}
\begin{tabular}{ccc}
\input{collinearityNew.pspdftex} &
\hspace*{2cm}&
\input{Grid.pspdftex}\\
\small (a)&&\small (b)
\end{tabular}
\caption{(a) The near collinearity that corresponds to a Delaunay edge
ceasing to be $\alpha$-stable.
(b) A set of points for which the number of $\alpha$-stable edges
in $\DT(P)$ (those corresponding to the vertical Voronoi edges)
is close to $n$.}
\label{collinearity}
\end{center}
\end{figure}

\paragraph{SDGs are not too sparse.}
Let $P$ be a set of $n$ points in the plane.  We give a lower bound
on the number of $\alpha$-stable Delaunay edges in the Delaunay
triangulation of $P$. Our lower bound approaches $n$ as $\alpha$
decreases to zero.

Let $n_0$ be the number of points with no incident $\alpha$-stable
edges in $\DT(P)$ and let $n_1$ be the number of points with a
single incident $\alpha$-stable edge in $\DT(P)$. Clearly the total
number of $\alpha$-stable edges in $\DT(P)$ is at least
\begin{equation} \label{eqnZ}
 \frac{2(n-n_0 -n_1) +
n_1}{2} = n - \left( \frac{2n_0 + n_1}{2} \right) \ .\end{equation}

We now derive an upper bound on $\frac{2n_0 + n_1}{2}$. Consider a
vertex $v$ with no incident $\alpha$-stable edges. If $v$ is not a
vertex of the convex hull then its degree in $\DT(P)$ must be at least $2\pi /
\alpha$ (the boundary of its cell in $\VD(P)$ contains at least
$2\pi / \alpha$ $\alpha$-short edges). If $v$ is a vertex of the
convex hull then its degree must be at least $(2\pi - d(v)) /
\alpha$ where $d(v)$ is the angle between the two infinite rays
bounding the Voronoi cell of $v$.  Similarly,
consider a vertex $v$ with one incident $\alpha$-stable edge. If $v$
is not a vertex of the convex hull then its degree must be at least
$\pi/\alpha + 1$ and if $v$ is a vertex of the hull then  its degree
is at least $(\pi-d(v))/\alpha + 1$. Since $\sum d(v)$ over all hull
vertices is $2\pi$, we  get that the sum of the degrees of the
vertices in $\DT(P)$ is at least
\begin{equation} \label{degreeLB}
n_1 \left(\frac{\pi}{\alpha}+1\right)+2n_0\frac{\pi}{\alpha} -
\frac{2\pi}{\alpha}  \ .
\end{equation}
On the other hand, the Delaunay triangulation of any set with $h$
vertices on the convex hull consists of $3n-h-3$ edges so the sum of
the degrees is $6n-2h-6$. Combining this with the lower bound in
(\ref{degreeLB}) we get that
$$
n_1 \left(\frac{\pi}{\alpha}+1\right)+2n_0\frac{\pi}{\alpha} -
\frac{2\pi}{\alpha} \le 6n ,
$$
which implies that
$$
\left( \frac{2n_0 + n_1}{2} \right) \le \frac{6n}{\pi/\alpha} + 2 \
.
$$
Substituting this upper bound in Equation (\ref{eqnZ}) we get that
the number of $\alpha$-stable edges in $\DT(P)$ is at least
$$
n \left(1 - \frac{6\alpha}{\pi} \right) - 2 \ .
$$

This is nearly tight, since, for any $\alpha$, there exist  sets of $n$ points for
which the number of $\alpha$-stable edges is  roughly $n$; see
Figure \ref{collinearity}(b).

\paragraph{Closest pairs, crusts, $\beta$-skeleta, and the SDG.}
Let $\beta\geq 1$, and let $P$ be a set of $n$ points in the plane.
The \textit{$\beta$-skeleton} of $P$ is a graph on $P$ that consists
of all the edges $pq$ such that the union of the two disks of radius
$(\beta/2)d(p,q)$, touching $p$ and $q$, does not contain any point
of $P\setminus\{p,q\}$. See, e.g., \cite{Crusts,Skeletons} for
properties of the $\beta$-skeleton, and for its applications in
surface reconstruction.  We claim that the edges of the
$\beta$-skeleton are $\alpha$-stable in $\DT(P)$, provided
$\beta\geq 1+\Omega(\alpha^2)$. Indeed, let $pq$ be an edge of the
$\beta$-skeleton of $P$, for $\beta>1$.
 Let $c_1$ and $c_2$ be the centers of the two empty disks
of radius $(\beta/2)d(p,q)$ touching $p$ and $q$; see
Figure~\ref{Fig:Skeleton}(a). Clearly $\angle c_1pq=\angle c_2pq$.
Denote $\theta =  \angle c_1pq=\angle c_2pq $. Each of $p,q$ sees
the Voronoi edge $e_{pq}$ at an angle at least $2\theta$, so it is
$2\theta$-stable. We have $1/\beta=\cos \theta\approx 1-\theta^2/2$
or $\beta=1+\Theta(\theta^2)$. That is, for $\beta\geq
1+\Omega(\alpha^2)$, with an appropriate (small) constant of
proportionality, $pq$ is $\alpha$-stable.

\begin{figure}[htbp]
\centering
\begin{tabular}{ccccc}
\input{Skeleton.pspdftex}&\hspace*{5mm}&
\input{norng.pspdftex}&\hspace*{5mm}&
\input{wheel.pspdftex}\\
\small (a) && \small (b)&& \small (c)
\end{tabular}
\caption {(a) The skeleton edge $pq$ is a stable edge. (b) An  edge
$ab$ of the Relative Neighborhood Graph that is not stable. (c)
Point $p$ is disconnected in the SDG. $\VD$ is drawn with dashed
lines, $\SDG$ with solid lines, and the remaining $\DT$ edges with
dotted lines.} \label{Fig:Skeleton}
\end{figure}
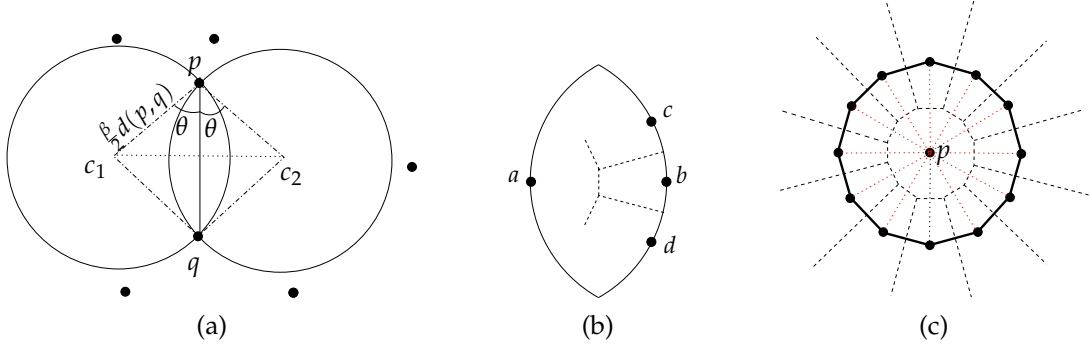

A similar argument shows that the stable Delaunay graph contains the
closest pair in $P(t)$ as well as the {\it crust} of a set of points
sampled sufficiently densely along a 1-dimensional curve (see
\cite{Amenta,Crusts} for the definition of crusts and their
applications in surface reconstruction). We  sketch the argument for
closest pairs: If $(p,q)$ is a closest pair then $pq\in \DT(P)$, and
the two adjacent Delaunay triangles $\triangle pqr^+,\triangle
pqr^-$ are such that their angles at $r^+,r^-$ are at most $\pi/3$
each, so $e_{pq}$ is $(\pi/3)$-long, ensuring that $pq$ belongs to
any stable subgraph for $\alpha$ sufficiently small; see
\cite{KineticNeighbors} for more details. We omit the proof for
crusts, which is fairly straightforward too.

\paragraph{Remark.} 
%
Stable Delaunay graphs need not contain all the edges
of several other important subgraphs of the Delaunay triangulation,
including the Euclidean minimum spanning tree, the Gabriel graph,
the relative neighborhood graph, and the all-nearest-neighbors
graph. An illustration for the relative neighborhood graph (RNG) is
given in Figure~\ref{Fig:Skeleton}~(b). 
Recall that an edge $pq$ is in RNG if there is no point $r\in P$ 
such that $\max\{\|pr\|, \|qr\|\} < \|pq\|$. As shown in figure that 
$pq$ is an edge of RNG, but the angular extent of the dual Voronoi 
edge $e_{pq}$ can be arbitrarily small.
As a matter of fact, the stable Delaunay graph
need not even be connected, as is illustrated in
Figure~\ref{Fig:Skeleton}(c).

\section{Conclusion}
\label{sec:concl}

In this paper we introduced the notion of a stable Delaunay graph
(SDG), a large subgraph of the Delaunay triangulation, which retains
several useful properties of the full Delaunay triangulation. We proved
that a $4\alpha$-stable edge in (the Euclidean) $\DT(P)$ is $\alpha$-stable 
in $\DT^Q(P)$, where $Q$ is a regular $k$-gon for any
$k \ge 2\pi/\alpha$ and $\alpha$ is the stability parameter, and
that the dual Voronoi edge $e_{pq}^\poly$ in $\VD^\poly(P)$ contains
at least eleven breakpoints. 
Using these properties and the kinetic data structure for $\DT^\poly(P)$
developed in the companion paper~\cite{polydel}, we presented a 
linear-size KDS for maintaining a Euclidean $(8\alpha,\alpha)$-SDG 
as the input points move. The KDS processes only a nearly quadratic 
number of events 
if the points move along algebraic trajectories of bounded degree, and each
event can be processed in $O(\log n)$ time. We also showed that if
an edge is stable in the Delaunay triangulation under the Euclidean
norm, it is also stable in the Delaunay triangulation under any 
convex distance function sufficiently close to the Euclidean
norm, and vice versa.

Proving a subcubic upper bound on the number of topological changes
in the Euclidean Delaunay triangulation for a set of moving points
still remains elusive (in spite of the recent progress 
in~\cite{Rubin, RubinUnit}), but our result implies that if
the true bound is really close to cubic, or just significantly
super-quadratic, then the overwhelming majority of these changes
involve  edges appearing and disappearing while  the four vertices
of the two triangles adjacent to each such edge remain nearly
cocircular throughout the entire time in which the edge exists.

We conclude by mentioning two open problems:
\begin{itemize}
%
\item[(i)] Is there a KDS for maintaining a triangulation of the 
\emph{entire} convex hull of a set of moving points in the plane,
which is an \emph{approximate} Delaunay triangulation, defined
appropriately,
and which processes only a near-quadratic number of events? 
In particular, can the SDG maintained by our KDS be extended to a
triangulation scheme of $\conv (P)$ (recall Figure~\ref{fig:VDT}(b)), 
e.g., using the ideas from the kinetic triangulation schemes presented 
in~\cite{AWY,KRS}, which also undergoes only a near-quadratic number of 
topological changes during the motion?  Perhaps a deeper analysis of the
structure of the ``holes" in the stable sub-diagram may yield a solution
to this problem, using the fact that for every missing edge, the two
incident triangles form a near-cocircularity in the diagram. This may
lead to a scheme that fills in the holes by
near-Delaunay edges and that has the desired properties.

\item[(ii)] What are the other large and interesting subgraphs of
$\DT(P)$ that undergo only a near-quadratic number of topological 
changes under a motion of the points of $P$ of the above kind? 
For instance, can one prove that there are only a near-quadratic number
of changes in the $\alpha$-shape or the relative neighborhood graph of $P$ if the points 
of $P$ move along algebraic trajectories of bounded degree.
\end{itemize}

\paragraph{Acknowledgements.}
{\small
P.A. and M.S. were supported by Grant 2012/229 from the U.S.-Israel 
Binational Science Foundation.  P.A. was also supported by
NSF under grants CCF-09-40671, CCF-10-12254, and CCF-11-61359, and by an ERDC
contract W9132V-11-C-0003.
L.G. was supported by NSF grants CCF-10-11228 and CCF-11-61480.
H.K. was supported by grant 822/10 from the
      Israel Science Foundation, grant 1161/2011 from the
      German-Israeli Science Foundation, and by the Israeli Centers
      for Research Excellence (I-CORE) program (center no.~4/11). %
N.R. was supported by Grants 975/06 and 338/09 from the
Israel Science Fund, by Minerva Fellowship Program of the Max Planck Society, by the Fondation Sciences Math\'{e}matiques de Paris (FSMP), and by a public grant overseen by the French National Research Agency (ANR) as part of the "Investissements d'Avenir" program (reference: ANR-10-LABX-0098).
M.S. was supported by NSF Grant CCF-08-30272, by Grants 338/09 and 892/13 from the Israel Science Foundation, by the Israeli Centers for Research Excellence (I-CORE) program (center no.~4/11), and by the Hermann Minkowski--MINERVA Center for Geometry at Tel Aviv University.
}
\end{document}

%% file: LongDelaunay.pspdftex
\begin{picture}(0,0)%
\includegraphics{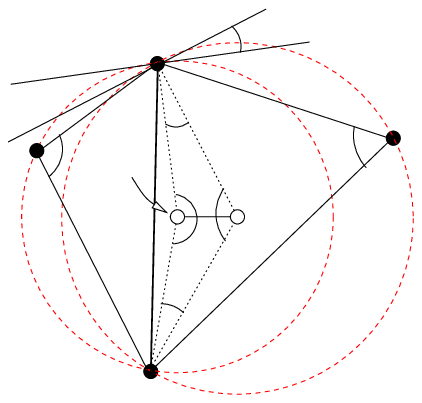}%
\end{picture}%
\setlength{\unitlength}{1973sp}%
\begingroup\makeatletter\ifx\SetFigFont\undefined%
\gdef\SetFigFont#1#2#3#4#5{%
  \reset@font\fontsize{#1}{#2pt}%
  \fontfamily{#3}\fontseries{#4}\fontshape{#5}%
  \selectfont}%
\fi\endgroup%
\begin{picture}(4062,3927)(990,-3654)
\put(1236,-2072){\makebox(0,0)[rb]{\smash{{\SetFigFont{10}{12.0}{\rmdefault}{\mddefault}{\updefault}{\color[rgb]{0,0,0}$D^-$}%
}}}}
\put(5001,-2632){\makebox(0,0)[lb]{\smash{{\SetFigFont{10}{12.0}{\rmdefault}{\mddefault}{\updefault}{\color[rgb]{0,0,0}$D^+$}%
}}}}
\put(2322,-111){\makebox(0,0)[lb]{\smash{{\SetFigFont{10}{12.0}{\rmdefault}{\mddefault}{\updefault}{\color[rgb]{0,0,0}$p$}%
}}}}
\put(2265,-3545){\makebox(0,0)[lb]{\smash{{\SetFigFont{10}{12.0}{\rmdefault}{\mddefault}{\updefault}{\color[rgb]{0,0,0}$q$}%
}}}}
\put(3457,-63){\makebox(0,0)[lb]{\smash{{\SetFigFont{10}{12.0}{\rmdefault}{\mddefault}{\updefault}{\color[rgb]{0,0,0}$\alpha$}%
}}}}
\put(2917,-873){\makebox(0,0)[lb]{\smash{{\SetFigFont{10}{12.0}{\rmdefault}{\mddefault}{\updefault}{\color[rgb]{0,0,0}$\alpha$}%
}}}}
\put(3486,-1865){\makebox(0,0)[lb]{\smash{{\SetFigFont{10}{12.0}{\rmdefault}{\mddefault}{\updefault}{\color[rgb]{0,0,0}$b$}%
}}}}
\put(2679,-2480){\makebox(0,0)[lb]{\smash{{\SetFigFont{10}{12.0}{\rmdefault}{\mddefault}{\updefault}{\color[rgb]{0,0,0}$\alpha$}%
}}}}
\put(4965,-1003){\makebox(0,0)[lb]{\smash{{\SetFigFont{10}{12.0}{\rmdefault}{\mddefault}{\updefault}{\color[rgb]{0,0,0}$r^+$}%
}}}}
\put(1005,-1313){\makebox(0,0)[lb]{\smash{{\SetFigFont{10}{12.0}{\rmdefault}{\mddefault}{\updefault}{\color[rgb]{0,0,0}$r^-$}%
}}}}
\put(2119,-1295){\makebox(0,0)[lb]{\smash{{\SetFigFont{10}{12.0}{\rmdefault}{\mddefault}{\updefault}{\color[rgb]{0,0,0}$a$}%
}}}}
\end{picture}%

%% file: normasis.pspdftex
\begin{picture}(0,0)%
\includegraphics{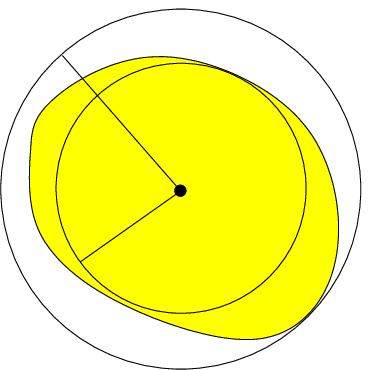}%
\end{picture}%
\setlength{\unitlength}{1973sp}%
\begingroup\makeatletter\ifx\SetFigFont\undefined%
\gdef\SetFigFont#1#2#3#4#5{%
  \reset@font\fontsize{#1}{#2pt}%
  \fontfamily{#3}\fontseries{#4}\fontshape{#5}%
  \selectfont}%
\fi\endgroup%
\begin{picture}(3472,3470)(1494,-4143)
\put(3347,-2459){\makebox(0,0)[lb]{\smash{{\SetFigFont{10}{12.0}{\rmdefault}{\mddefault}{\updefault}{\color[rgb]{0,0,0}$o$}%
}}}}
\put(2116,-2626){\makebox(0,0)[lb]{\smash{{\SetFigFont{10}{12.0}{\rmdefault}{\mddefault}{\updefault}{\color[rgb]{0,0,0}$D_I$}%
}}}}
\put(2701,-3999){\makebox(0,0)[lb]{\smash{{\SetFigFont{10}{12.0}{\rmdefault}{\mddefault}{\updefault}{\color[rgb]{0,0,0}$D_O$}%
}}}}
\put(2708,-1740){\makebox(0,0)[lb]{\smash{{\SetFigFont{10}{12.0}{\rmdefault}{\mddefault}{\updefault}{\color[rgb]{0,0,0}$1$}%
}}}}
\put(2671,-2979){\makebox(0,0)[lb]{\smash{{\SetFigFont{10}{12.0}{\rmdefault}{\mddefault}{\updefault}{\color[rgb]{0,0,0}$\cos\alpha$}%
}}}}
\put(3729,-3782){\makebox(0,0)[lb]{\smash{{\SetFigFont{10}{12.0}{\rmdefault}{\mddefault}{\updefault}{\color[rgb]{0,0,0}$Q$}%
}}}}
\end{picture}%

%% file: Qcopies.pspdftex
\begin{picture}(0,0)%
\includegraphics{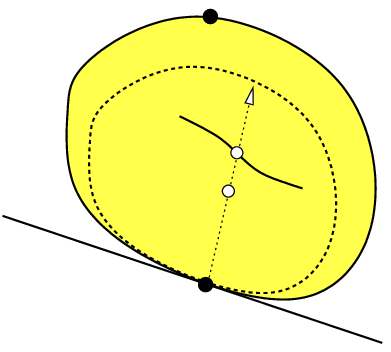}%
\end{picture}%
\setlength{\unitlength}{2368sp}%
\begingroup\makeatletter\ifx\SetFigFont\undefined%
\gdef\SetFigFont#1#2#3#4#5{%
  \reset@font\fontsize{#1}{#2pt}%
  \fontfamily{#3}\fontseries{#4}\fontshape{#5}%
  \selectfont}%
\fi\endgroup%
\begin{picture}(3127,3162)(2200,-3760)
\put(3779,-817){\makebox(0,0)[lb]{\smash{{\SetFigFont{10}{12.0}{\rmdefault}{\mddefault}{\updefault}{\color[rgb]{0,0,0}$q$}%
}}}}
\put(3323,-1313){\makebox(0,0)[lb]{\smash{{\SetFigFont{10}{12.0}{\rmdefault}{\mddefault}{\updefault}{\color[rgb]{0,0,0}$Q_{pq}(u)$}%
}}}}
\put(5312,-3669){\makebox(0,0)[lb]{\smash{{\SetFigFont{10}{12.0}{\rmdefault}{\mddefault}{\updefault}{\color[rgb]{0,0,0}$\ell_p(u)$}%
}}}}
\put(3644,-3425){\makebox(0,0)[lb]{\smash{{\SetFigFont{10}{12.0}{\rmdefault}{\mddefault}{\updefault}{\color[rgb]{0,0,0}$p$}%
}}}}
\put(4102,-2536){\makebox(0,0)[lb]{\smash{{\SetFigFont{10}{12.0}{\rmdefault}{\mddefault}{\updefault}{\color[rgb]{0,0,0}$c$}%
}}}}
\put(4224,-1836){\makebox(0,0)[lb]{\smash{{\SetFigFont{10}{12.0}{\rmdefault}{\mddefault}{\updefault}{\color[rgb]{0,0,0}$u[p]$}%
}}}}
\put(4216,-2107){\makebox(0,0)[lb]{\smash{{\SetFigFont{10}{12.0}{\rmdefault}{\mddefault}{\updefault}{\color[rgb]{0,0,0}$c_{pq}$}%
}}}}
\put(2963,-2214){\makebox(0,0)[lb]{\smash{{\SetFigFont{10}{12.0}{\rmdefault}{\mddefault}{\updefault}{\color[rgb]{0,0,0}$Q_{p,c}(u)$}%
}}}}
\put(3473,-1710){\makebox(0,0)[lb]{\smash{{\SetFigFont{10}{12.0}{\rmdefault}{\mddefault}{\updefault}{\color[rgb]{0,0,0}$b_{pq}^Q$}%
}}}}
\end{picture}%

%% file: TopologicQ.pspdftex
\begin{picture}(0,0)%
\includegraphics{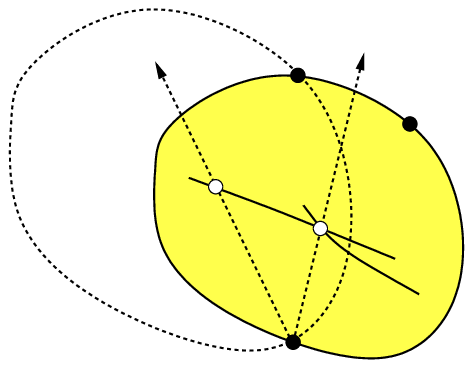}%
\end{picture}%
\setlength{\unitlength}{2368sp}%
\begingroup\makeatletter\ifx\SetFigFont\undefined%
\gdef\SetFigFont#1#2#3#4#5{%
  \reset@font\fontsize{#1}{#2pt}%
  \fontfamily{#3}\fontseries{#4}\fontshape{#5}%
  \selectfont}%
\fi\endgroup%
\begin{picture}(4018,3089)(1229,-3521)
\put(3644,-3425){\makebox(0,0)[lb]{\smash{{\SetFigFont{10}{12.0}{\rmdefault}{\mddefault}{\updefault}{\color[rgb]{0,0,0}$p$}%
}}}}
\put(4484,-2995){\makebox(0,0)[lb]{\smash{{\SetFigFont{10}{12.0}{\rmdefault}{\mddefault}{\updefault}{\color[rgb]{0,0,0}$b^Q_{pq}$}%
}}}}
\put(3845,-874){\makebox(0,0)[lb]{\smash{{\SetFigFont{10}{12.0}{\rmdefault}{\mddefault}{\updefault}{\color[rgb]{0,0,0}$r$}%
}}}}
\put(4166,-736){\makebox(0,0)[lb]{\smash{{\SetFigFont{10}{12.0}{\rmdefault}{\mddefault}{\updefault}{\color[rgb]{0,0,0}$u^*[p]$}%
}}}}
\put(4711,-1218){\makebox(0,0)[lb]{\smash{{\SetFigFont{10}{12.0}{\rmdefault}{\mddefault}{\updefault}{\color[rgb]{0,0,0}$q$}%
}}}}
\put(5232,-2169){\makebox(0,0)[lb]{\smash{{\SetFigFont{10}{12.0}{\rmdefault}{\mddefault}{\updefault}{\color[rgb]{0,0,0}$Q_{pr}(u^*)$}%
}}}}
\put(3382,-1907){\makebox(0,0)[lb]{\smash{{\SetFigFont{10}{12.0}{\rmdefault}{\mddefault}{\updefault}{\color[rgb]{0,0,0}$b^Q_{pr}$}%
}}}}
\put(3673,-2356){\makebox(0,0)[lb]{\smash{{\SetFigFont{10}{12.0}{\rmdefault}{\mddefault}{\updefault}{\color[rgb]{0,0,0}$c^*$}%
}}}}
\put(2371,-786){\makebox(0,0)[lb]{\smash{{\SetFigFont{10}{12.0}{\rmdefault}{\mddefault}{\updefault}{\color[rgb]{0,0,0}$u[p]$}%
}}}}
\put(1244,-2933){\makebox(0,0)[lb]{\smash{{\SetFigFont{10}{12.0}{\rmdefault}{\mddefault}{\updefault}{\color[rgb]{0,0,0}$Q_{pr}(u)$}%
}}}}
\end{picture}%

%% file: convexkgon.pspdftex
\begin{picture}(0,0)%
\includegraphics{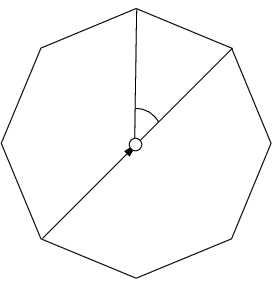}%
\end{picture}%
\setlength{\unitlength}{1973sp}%
\begingroup\makeatletter\ifx\SetFigFont\undefined%
\gdef\SetFigFont#1#2#3#4#5{%
  \reset@font\fontsize{#1}{#2pt}%
  \fontfamily{#3}\fontseries{#4}\fontshape{#5}%
  \selectfont}%
\fi\endgroup%
\begin{picture}(2637,3223)(2791,-2743)
\put(4213,-1093){\makebox(0,0)[lb]{\smash{{\SetFigFont{10}{12.0}{\rmdefault}{\mddefault}{\updefault}{\color[rgb]{0,0,0}$o$}%
}}}}
\put(5113,-93){\makebox(0,0)[lb]{\smash{{\SetFigFont{10}{12.0}{\rmdefault}{\mddefault}{\updefault}{\color[rgb]{0,0,0}$v_0$}%
}}}}
\put(4123,237){\makebox(0,0)[lb]{\smash{{\SetFigFont{10}{12.0}{\rmdefault}{\mddefault}{\updefault}{\color[rgb]{0,0,0}$v_1$}%
}}}}
\put(3213,-423){\makebox(0,0)[lb]{\smash{{\SetFigFont{10}{12.0}{\rmdefault}{\mddefault}{\updefault}{\color[rgb]{0,0,0}$v_2$}%
}}}}
\put(2873,-1193){\makebox(0,0)[lb]{\smash{{\SetFigFont{10}{12.0}{\rmdefault}{\mddefault}{\updefault}{\color[rgb]{0,0,0}$v_3$}%
}}}}
\put(4123,-2643){\makebox(0,0)[lb]{\smash{{\SetFigFont{10}{12.0}{\rmdefault}{\mddefault}{\updefault}{\color[rgb]{0,0,0}$v_5$}%
}}}}
\put(5053,-2213){\makebox(0,0)[lb]{\smash{{\SetFigFont{10}{12.0}{\rmdefault}{\mddefault}{\updefault}{\color[rgb]{0,0,0}$v_6$}%
}}}}
\put(5413,-1093){\makebox(0,0)[lb]{\smash{{\SetFigFont{10}{12.0}{\rmdefault}{\mddefault}{\updefault}{\color[rgb]{0,0,0}$v_7$}%
}}}}
\put(3913,-1463){\makebox(0,0)[lb]{\smash{{\SetFigFont{10}{12.0}{\rmdefault}{\mddefault}{\updefault}{\color[rgb]{0,0,0}$u_4$}%
}}}}
\put(4263,-683){\makebox(0,0)[lb]{\smash{{\SetFigFont{10}{12.0}{\rmdefault}{\mddefault}{\updefault}{\color[rgb]{0,0,0}$\alpha$}%
}}}}
\put(3109,-2213){\makebox(0,0)[rb]{\smash{{\SetFigFont{10}{12.0}{\rmdefault}{\mddefault}{\updefault}{\color[rgb]{0,0,0}$v_4$}%
}}}}
\put(2973,-623){\makebox(0,0)[rb]{\smash{{\SetFigFont{10}{12.0}{\rmdefault}{\mddefault}{\updefault}{\color[rgb]{0,0,0}$Q$}%
}}}}
\end{picture}%

%% file: bisect.pspdftex
\begin{picture}(0,0)%
\includegraphics{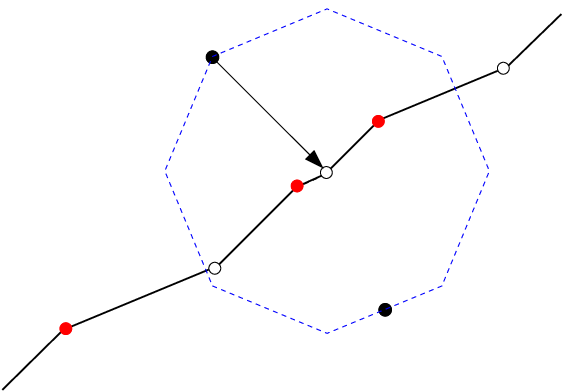}%
\end{picture}%
\setlength{\unitlength}{1973sp}%
\begingroup\makeatletter\ifx\SetFigFont\undefined%
\gdef\SetFigFont#1#2#3#4#5{%
  \reset@font\fontsize{#1}{#2pt}%
  \fontfamily{#3}\fontseries{#4}\fontshape{#5}%
  \selectfont}%
\fi\endgroup%
\begin{picture}(5409,3691)(1222,-3462)
\put(5026,-3001){\makebox(0,0)[lb]{\smash{{\SetFigFont{10}{12.0}{\rmdefault}{\mddefault}{\updefault}{\color[rgb]{0,0,0}$q$}%
}}}}
\put(4530,-1575){\makebox(0,0)[lb]{\smash{{\SetFigFont{10}{12.0}{\rmdefault}{\mddefault}{\updefault}{\color[rgb]{0,0,0}$\bisect_{pq}^\poly$}%
}}}}
\put(2984,-49){\makebox(0,0)[lb]{\smash{{\SetFigFont{10}{12.0}{\rmdefault}{\mddefault}{\updefault}{\color[rgb]{0,0,0}$p$}%
}}}}
\end{picture}%

%% file: Euclphi.pspdftex
\begin{picture}(0,0)%
\includegraphics{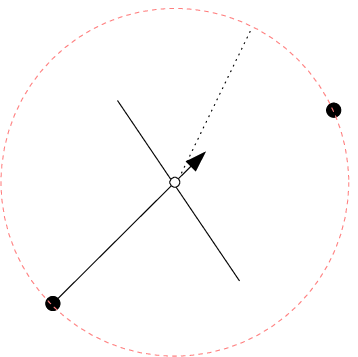}%
\end{picture}%
\setlength{\unitlength}{1973sp}%
\begingroup\makeatletter\ifx\SetFigFont\undefined%
\gdef\SetFigFont#1#2#3#4#5{%
  \reset@font\fontsize{#1}{#2pt}%
  \fontfamily{#3}\fontseries{#4}\fontshape{#5}%
  \selectfont}%
\fi\endgroup%
\begin{picture}(3354,3354)(1396,-4234)
\put(1488,-4059){\makebox(0,0)[lb]{\smash{{\SetFigFont{10}{12.0}{\rmdefault}{\mddefault}{\updefault}{\color[rgb]{0,0,0}$p$}%
}}}}
\put(4259,-2139){\makebox(0,0)[lb]{\smash{{\SetFigFont{10}{12.0}{\rmdefault}{\mddefault}{\updefault}{\color[rgb]{0,0,0}$q$}%
}}}}
\put(2434,-3334){\makebox(0,0)[lb]{\smash{{\SetFigFont{10}{12.0}{\rmdefault}{\mddefault}{\updefault}{\color[rgb]{0,0,0}$u_j[p]$}%
}}}}
\put(3301,-2885){\makebox(0,0)[lb]{\smash{{\SetFigFont{10}{12.0}{\rmdefault}{\mddefault}{\updefault}{\color[rgb]{0,0,0}$\bisect_{pq}$}%
}}}}
\put(3159,-2292){\rotatebox{65.0}{\makebox(0,0)[lb]{\smash{{\SetFigFont{10}{12.0}{\rmdefault}{\mddefault}{\updefault}{\color[rgb]{0,0,0}$\varphi_j[p,q]$}%
}}}}}
\end{picture}%

%% file: PolygonalDist.pspdftex
\begin{picture}(0,0)%
\includegraphics{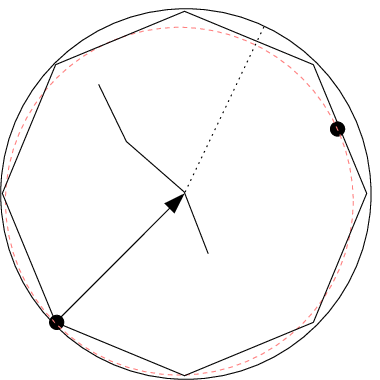}%
\end{picture}%
\setlength{\unitlength}{1973sp}%
\begingroup\makeatletter\ifx\SetFigFont\undefined%
\gdef\SetFigFont#1#2#3#4#5{%
  \reset@font\fontsize{#1}{#2pt}%
  \fontfamily{#3}\fontseries{#4}\fontshape{#5}%
  \selectfont}%
\fi\endgroup%
\begin{picture}(3572,3574)(1356,-4276)
\put(1488,-4059){\makebox(0,0)[lb]{\smash{{\SetFigFont{10}{12.0}{\rmdefault}{\mddefault}{\updefault}{\color[rgb]{0,0,0}$p$}%
}}}}
\put(4259,-2139){\makebox(0,0)[lb]{\smash{{\SetFigFont{10}{12.0}{\rmdefault}{\mddefault}{\updefault}{\color[rgb]{0,0,0}$q$}%
}}}}
\put(3481,-3085){\makebox(0,0)[lb]{\smash{{\SetFigFont{10}{12.0}{\rmdefault}{\mddefault}{\updefault}{\color[rgb]{0,0,0}$\bisect^\poly_{pq}$}%
}}}}
\put(2434,-3334){\makebox(0,0)[lb]{\smash{{\SetFigFont{10}{12.0}{\rmdefault}{\mddefault}{\updefault}{\color[rgb]{0,0,0}$u_j[p]$}%
}}}}
\put(3069,-2392){\rotatebox{65.0}{\makebox(0,0)[lb]{\smash{{\SetFigFont{10}{12.0}{\rmdefault}{\mddefault}{\updefault}{\color[rgb]{0,0,0}$\varphi^\poly_j[p,q]$}%
}}}}}
\end{picture}%

%% file: UndefinedPolygDist.pspdftex
\begin{picture}(0,0)%
\includegraphics{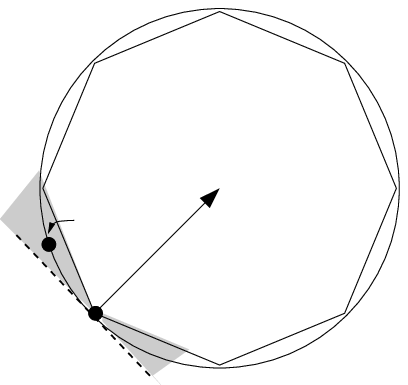}%
\end{picture}%
\setlength{\unitlength}{1973sp}%
\begingroup\makeatletter\ifx\SetFigFont\undefined%
\gdef\SetFigFont#1#2#3#4#5{%
  \reset@font\fontsize{#1}{#2pt}%
  \fontfamily{#3}\fontseries{#4}\fontshape{#5}%
  \selectfont}%
\fi\endgroup%
\begin{picture}(3843,3620)(966,-4273)
\put(2401,-3211){\makebox(0,0)[lb]{\smash{{\SetFigFont{10}{12.0}{\rmdefault}{\mddefault}{\updefault}{\color[rgb]{0,0,0}$u_j[p]$}%
}}}}
\put(2051,-3611){\makebox(0,0)[lb]{\smash{{\SetFigFont{10}{12.0}{\rmdefault}{\mddefault}{\updefault}{\color[rgb]{0,0,0}$p$}%
}}}}
\put(1656,-2822){\makebox(0,0)[lb]{\smash{{\SetFigFont{10}{12.0}{\rmdefault}{\mddefault}{\updefault}{\color[rgb]{0,0,0}$q$}%
}}}}
\end{picture}%

%% file: ProvePolyg1.pspdftex
\begin{picture}(0,0)%
\includegraphics{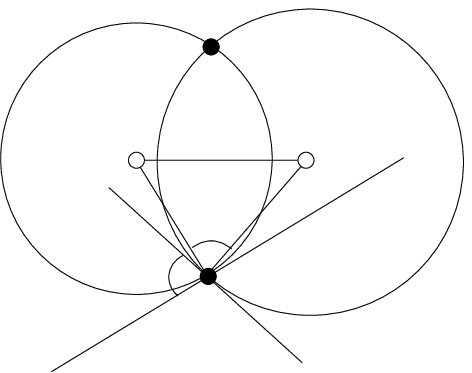}%
\end{picture}%
\setlength{\unitlength}{1973sp}%
\begingroup\makeatletter\ifx\SetFigFont\undefined%
\gdef\SetFigFont#1#2#3#4#5{%
  \reset@font\fontsize{#1}{#2pt}%
  \fontfamily{#3}\fontseries{#4}\fontshape{#5}%
  \selectfont}%
\fi\endgroup%
\begin{picture}(4458,3598)(1333,-3414)
\put(2381,-1234){\makebox(0,0)[lb]{\smash{{\SetFigFont{10}{12.0}{\rmdefault}{\mddefault}{\updefault}{\color[rgb]{0,0,0}$w_1$}%
}}}}
\put(3194,-2850){\makebox(0,0)[lb]{\smash{{\SetFigFont{10}{12.0}{\rmdefault}{\mddefault}{\updefault}{\color[rgb]{0,0,0}$p$}%
}}}}
\put(2594,-2589){\makebox(0,0)[lb]{\smash{{\SetFigFont{10}{12.0}{\rmdefault}{\mddefault}{\updefault}{\color[rgb]{0,0,0}$\beta$}%
}}}}
\put(4097,-1248){\makebox(0,0)[lb]{\smash{{\SetFigFont{10}{12.0}{\rmdefault}{\mddefault}{\updefault}{\color[rgb]{0,0,0}$w_2$}%
}}}}
\put(3207,-2086){\makebox(0,0)[lb]{\smash{{\SetFigFont{10}{12.0}{\rmdefault}{\mddefault}{\updefault}{\color[rgb]{0,0,0}$\beta$}%
}}}}
\put(3217,-83){\makebox(0,0)[lb]{\smash{{\SetFigFont{10}{12.0}{\rmdefault}{\mddefault}{\updefault}{\color[rgb]{0,0,0}$q$}%
}}}}
\put(1436,-1145){\makebox(0,0)[lb]{\smash{{\SetFigFont{10}{12.0}{\rmdefault}{\mddefault}{\updefault}{\color[rgb]{0,0,0}$D_1$}%
}}}}
\put(5417,-878){\makebox(0,0)[rb]{\smash{{\SetFigFont{10}{12.0}{\rmdefault}{\mddefault}{\updefault}{\color[rgb]{0,0,0}$D_2$}%
}}}}
\end{picture}%

%% file: ProvePolyg2.pspdftex
\begin{picture}(0,0)%
\includegraphics{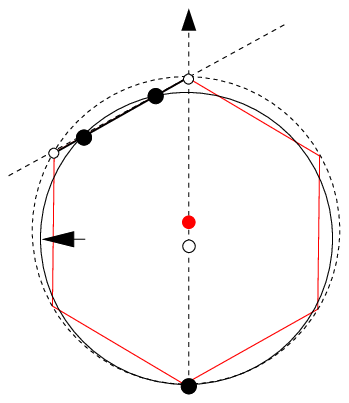}%
\end{picture}%
\setlength{\unitlength}{1973sp}%
\begingroup\makeatletter\ifx\SetFigFont\undefined%
\gdef\SetFigFont#1#2#3#4#5{%
  \reset@font\fontsize{#1}{#2pt}%
  \fontfamily{#3}\fontseries{#4}\fontshape{#5}%
  \selectfont}%
\fi\endgroup%
\begin{picture}(3290,4471)(1103,-4281)
\put(2828,-77){\makebox(0,0)[lb]{\smash{{\SetFigFont{10}{12.0}{\rmdefault}{\mddefault}{\updefault}{\color[rgb]{0,0,0}$u_j[p]$}%
}}}}
\put(3797,-1100){\makebox(0,0)[lb]{\smash{{\SetFigFont{10}{12.0}{\rmdefault}{\mddefault}{\updefault}{\color[rgb]{1,0,0}$D^+$}%
}}}}
\put(2861,-4172){\makebox(0,0)[lb]{\smash{{\SetFigFont{10}{12.0}{\rmdefault}{\mddefault}{\updefault}{\color[rgb]{0,0,0}$p$}%
}}}}
\put(2562,-1346){\makebox(0,0)[lb]{\smash{{\SetFigFont{10}{12.0}{\rmdefault}{\mddefault}{\updefault}{\color[rgb]{0,0,0}$q'$}%
}}}}
\put(1822,-1743){\makebox(0,0)[lb]{\smash{{\SetFigFont{10}{12.0}{\rmdefault}{\mddefault}{\updefault}{\color[rgb]{0,0,0}$q$}%
}}}}
\put(2183,-1485){\makebox(0,0)[lb]{\smash{{\SetFigFont{10}{12.0}{\rmdefault}{\mddefault}{\updefault}{\color[rgb]{0,0,0}$e'$}%
}}}}
\put(1118,-1524){\makebox(0,0)[lb]{\smash{{\SetFigFont{10}{12.0}{\rmdefault}{\mddefault}{\updefault}{\color[rgb]{0,0,0}$\ell'$}%
}}}}
\put(1967,-2497){\makebox(0,0)[lb]{\smash{{\SetFigFont{10}{12.0}{\rmdefault}{\mddefault}{\updefault}{\color[rgb]{0,0,0}$D$}%
}}}}
\put(3697,-3239){\makebox(0,0)[rb]{\smash{{\SetFigFont{10}{12.0}{\rmdefault}{\mddefault}{\updefault}{\color[rgb]{0,0,0}$Q_j$}%
}}}}
\end{picture}%

%% file: ProvePolyg3.pspdftex
\begin{picture}(0,0)%
\includegraphics{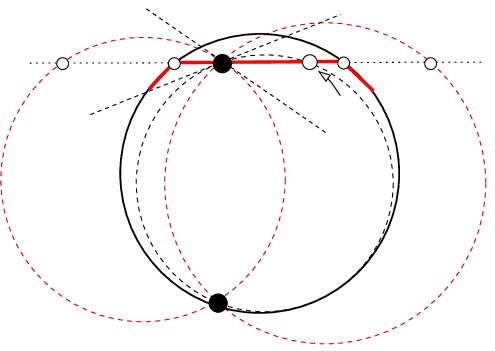}%
\end{picture}%
\setlength{\unitlength}{1973sp}%
\begingroup\makeatletter\ifx\SetFigFont\undefined%
\gdef\SetFigFont#1#2#3#4#5{%
  \reset@font\fontsize{#1}{#2pt}%
  \fontfamily{#3}\fontseries{#4}\fontshape{#5}%
  \selectfont}%
\fi\endgroup%
\begin{picture}(4671,3295)(738,-2946)
\put(953,-29){\makebox(0,0)[lb]{\smash{{\SetFigFont{10}{12.0}{\rmdefault}{\mddefault}{\updefault}{\color[rgb]{0,0,0}$q_1$}%
}}}}
\put(4105,-2475){\makebox(0,0)[lb]{\smash{{\SetFigFont{10}{12.0}{\rmdefault}{\mddefault}{\updefault}{\color[rgb]{0,0,0}$D^+$}%
}}}}
\put(4071,-84){\makebox(0,0)[lb]{\smash{{\SetFigFont{10}{12.0}{\rmdefault}{\mddefault}{\updefault}{\color[rgb]{0,0,0}$a_2$}%
}}}}
\put(1946,-143){\makebox(0,0)[lb]{\smash{{\SetFigFont{10}{12.0}{\rmdefault}{\mddefault}{\updefault}{\color[rgb]{0,0,0}$a_1$}%
}}}}
\put(2636,-2837){\makebox(0,0)[lb]{\smash{{\SetFigFont{10}{12.0}{\rmdefault}{\mddefault}{\updefault}{\color[rgb]{0,0,0}$p$}%
}}}}
\put(4816,-39){\makebox(0,0)[lb]{\smash{{\SetFigFont{10}{12.0}{\rmdefault}{\mddefault}{\updefault}{\color[rgb]{0,0,0}$q_2$}%
}}}}
\put(2753,-611){\makebox(0,0)[lb]{\smash{{\SetFigFont{10}{12.0}{\rmdefault}{\mddefault}{\updefault}{\color[rgb]{0,0,0}$q$}%
}}}}
\put(3912,-780){\makebox(0,0)[lb]{\smash{{\SetFigFont{10}{12.0}{\rmdefault}{\mddefault}{\updefault}{\color[rgb]{0,0,0}$q'$}%
}}}}
\put(1781,-939){\makebox(0,0)[rb]{\smash{{\SetFigFont{10}{12.0}{\rmdefault}{\mddefault}{\updefault}{\color[rgb]{0,0,0}$\tau$}%
}}}}
\put(5346,-1489){\makebox(0,0)[rb]{\smash{{\SetFigFont{10}{12.0}{\rmdefault}{\mddefault}{\updefault}{\color[rgb]{0,0,0}$D_2$}%
}}}}
\put(5116,-453){\makebox(0,0)[lb]{\smash{{\SetFigFont{10}{12.0}{\rmdefault}{\mddefault}{\updefault}{\color[rgb]{0,0,0}$\ell'$}%
}}}}
\put(942,-1947){\makebox(0,0)[lb]{\smash{{\SetFigFont{10}{12.0}{\rmdefault}{\mddefault}{\updefault}{\color[rgb]{0,0,0}$D_1$}%
}}}}
\put(3208,-435){\makebox(0,0)[lb]{\smash{{\SetFigFont{10}{12.0}{\rmdefault}{\mddefault}{\updefault}{\color[rgb]{0,0,0}$e'$}%
}}}}
\end{picture}%

%% file: Converse.pspdftex
\begin{picture}(0,0)%
\includegraphics{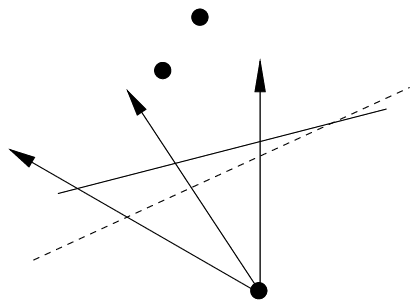}%
\end{picture}%
\setlength{\unitlength}{1973sp}%
\begingroup\makeatletter\ifx\SetFigFont\undefined%
\gdef\SetFigFont#1#2#3#4#5{%
  \reset@font\fontsize{#1}{#2pt}%
  \fontfamily{#3}\fontseries{#4}\fontshape{#5}%
  \selectfont}%
\fi\endgroup%
\begin{picture}(4212,3365)(1261,-3369)
\put(1276,-2969){\makebox(0,0)[lb]{\smash{{\SetFigFont{10}{12.0}{\rmdefault}{\mddefault}{\updefault}{\color[rgb]{0,0,0}$\bisect_{pr}$}%
}}}}
\put(1594,-2337){\makebox(0,0)[lb]{\smash{{\SetFigFont{10}{12.0}{\rmdefault}{\mddefault}{\updefault}{\color[rgb]{0,0,0}$\bisect_{pq}$}%
}}}}
\put(2659,-878){\makebox(0,0)[lb]{\smash{{\SetFigFont{10}{12.0}{\rmdefault}{\mddefault}{\updefault}{\color[rgb]{0,0,0}$r$}%
}}}}
\put(3198,-271){\makebox(0,0)[lb]{\smash{{\SetFigFont{10}{12.0}{\rmdefault}{\mddefault}{\updefault}{\color[rgb]{0,0,0}$q$}%
}}}}
\put(4121,-3260){\makebox(0,0)[lb]{\smash{{\SetFigFont{10}{12.0}{\rmdefault}{\mddefault}{\updefault}{\color[rgb]{0,0,0}$p$}%
}}}}
\put(4098,-1178){\makebox(0,0)[lb]{\smash{{\SetFigFont{10}{12.0}{\rmdefault}{\mddefault}{\updefault}{\color[rgb]{0,0,0}$u_j[p]$}%
}}}}
\put(1887,-1835){\makebox(0,0)[lb]{\smash{{\SetFigFont{10}{12.0}{\rmdefault}{\mddefault}{\updefault}{\color[rgb]{0,0,0}$u_{j-2}[p]$}%
}}}}
\put(2982,-1500){\makebox(0,0)[lb]{\smash{{\SetFigFont{10}{12.0}{\rmdefault}{\mddefault}{\updefault}{\color[rgb]{0,0,0}$u_{j-1}[p]$}%
}}}}
\end{picture}%

%% file: norm.pspdftex
\begin{picture}(0,0)%
\includegraphics{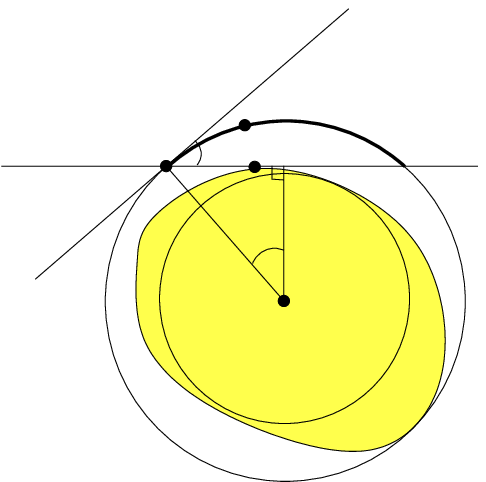}%
\end{picture}%
\setlength{\unitlength}{1973sp}%
\begingroup\makeatletter\ifx\SetFigFont\undefined%
\gdef\SetFigFont#1#2#3#4#5{%
  \reset@font\fontsize{#1}{#2pt}%
  \fontfamily{#3}\fontseries{#4}\fontshape{#5}%
  \selectfont}%
\fi\endgroup%
\begin{picture}(4599,4555)(501,-4162)
\put(886,-3278){\makebox(0,0)[lb]{\smash{{\SetFigFont{10}{12.0}{\rmdefault}{\mddefault}{\updefault}{\color[rgb]{0,0,0}$\bd D_O$}%
}}}}
\put(3347,-2459){\makebox(0,0)[lb]{\smash{{\SetFigFont{10}{12.0}{\rmdefault}{\mddefault}{\updefault}{\color[rgb]{0,0,0}$o$}%
}}}}
\put(2832,-1943){\makebox(0,0)[lb]{\smash{{\SetFigFont{10}{12.0}{\rmdefault}{\mddefault}{\updefault}{\color[rgb]{0,0,0}$\theta$}%
}}}}
\put(3939,-3632){\makebox(0,0)[lb]{\smash{{\SetFigFont{10}{12.0}{\rmdefault}{\mddefault}{\updefault}{\color[rgb]{0,0,0}$Q$}%
}}}}
\put(4562,-1060){\makebox(0,0)[lb]{\smash{{\SetFigFont{10}{12.0}{\rmdefault}{\mddefault}{\updefault}{\color[rgb]{0,0,0}$\ell$}%
}}}}
\put(3696,-710){\makebox(0,0)[lb]{\smash{{\SetFigFont{10}{12.0}{\rmdefault}{\mddefault}{\updefault}{\color[rgb]{0,0,0}$\gamma$}%
}}}}
\put(3090, 60){\makebox(0,0)[lb]{\smash{{\SetFigFont{10}{12.0}{\rmdefault}{\mddefault}{\updefault}{\color[rgb]{0,0,0}$\tau$}%
}}}}
\put(2886,-669){\makebox(0,0)[lb]{\smash{{\SetFigFont{10}{12.0}{\rmdefault}{\mddefault}{\updefault}{\color[rgb]{0,0,0}$y$}%
}}}}
\put(2442,-1108){\makebox(0,0)[lb]{\smash{{\SetFigFont{10}{12.0}{\rmdefault}{\mddefault}{\updefault}{\color[rgb]{0,0,0}$\theta$}%
}}}}
\put(2903,-1035){\makebox(0,0)[lb]{\smash{{\SetFigFont{10}{12.0}{\rmdefault}{\mddefault}{\updefault}{\color[rgb]{0,0,0}$x$}%
}}}}
\end{picture}%

%% file: ClaimNorm.pspdftex
\begin{picture}(0,0)%
\includegraphics{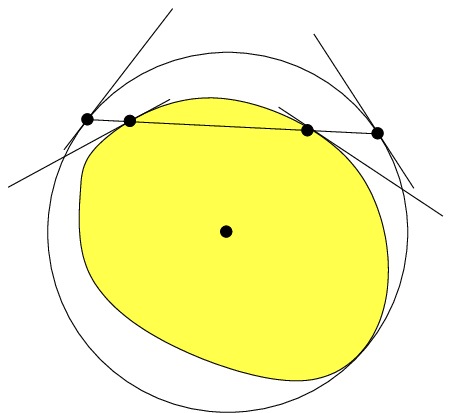}%
\end{picture}%
\setlength{\unitlength}{1973sp}%
\begingroup\makeatletter\ifx\SetFigFont\undefined%
\gdef\SetFigFont#1#2#3#4#5{%
  \reset@font\fontsize{#1}{#2pt}%
  \fontfamily{#3}\fontseries{#4}\fontshape{#5}%
  \selectfont}%
\fi\endgroup%
\begin{picture}(4608,4217)(727,-4165)
\put(3939,-3632){\makebox(0,0)[lb]{\smash{{\SetFigFont{10}{12.0}{\rmdefault}{\mddefault}{\updefault}{\color[rgb]{0,0,0}$Q$}%
}}}}
\put(742,-1911){\makebox(0,0)[lb]{\smash{{\SetFigFont{10}{12.0}{\rmdefault}{\mddefault}{\updefault}{\color[rgb]{0,0,0}$\ell_x$}%
}}}}
\put(2211,-1694){\makebox(0,0)[lb]{\smash{{\SetFigFont{10}{12.0}{\rmdefault}{\mddefault}{\updefault}{\color[rgb]{0,0,0}$x$}%
}}}}
\put(1449,-1281){\makebox(0,0)[lb]{\smash{{\SetFigFont{10}{12.0}{\rmdefault}{\mddefault}{\updefault}{\color[rgb]{0,0,0}$x'$}%
}}}}
\put(4755,-1445){\makebox(0,0)[lb]{\smash{{\SetFigFont{10}{12.0}{\rmdefault}{\mddefault}{\updefault}{\color[rgb]{0,0,0}$y'$}%
}}}}
\put(5320,-2427){\makebox(0,0)[lb]{\smash{{\SetFigFont{10}{12.0}{\rmdefault}{\mddefault}{\updefault}{\color[rgb]{0,0,0}$\ell_y$}%
}}}}
\put(3757,-1784){\makebox(0,0)[lb]{\smash{{\SetFigFont{10}{12.0}{\rmdefault}{\mddefault}{\updefault}{\color[rgb]{0,0,0}$y$}%
}}}}
\put(2663,-191){\makebox(0,0)[lb]{\smash{{\SetFigFont{10}{12.0}{\rmdefault}{\mddefault}{\updefault}{\color[rgb]{0,0,0}$\tau_1$}%
}}}}
\put(3903,-385){\makebox(0,0)[lb]{\smash{{\SetFigFont{10}{12.0}{\rmdefault}{\mddefault}{\updefault}{\color[rgb]{0,0,0}$\tau_2$}%
}}}}
\put(1080,-3299){\makebox(0,0)[lb]{\smash{{\SetFigFont{10}{12.0}{\rmdefault}{\mddefault}{\updefault}{\color[rgb]{0,0,0}$\partial D_O$}%
}}}}
\put(3347,-2459){\makebox(0,0)[lb]{\smash{{\SetFigFont{10}{12.0}{\rmdefault}{\mddefault}{\updefault}{\color[rgb]{0,0,0}$o$}%
}}}}
\end{picture}%

%% file: CloseTangents.pspdftex
\begin{picture}(0,0)%
\includegraphics{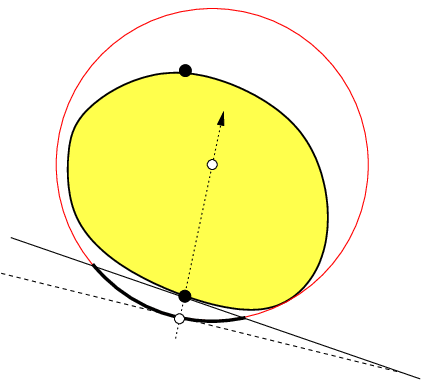}%
\end{picture}%
\setlength{\unitlength}{1973sp}%
\begingroup\makeatletter\ifx\SetFigFont\undefined%
\gdef\SetFigFont#1#2#3#4#5{%
  \reset@font\fontsize{#1}{#2pt}%
  \fontfamily{#3}\fontseries{#4}\fontshape{#5}%
  \selectfont}%
\fi\endgroup%
\begin{picture}(4044,3953)(2105,-3952)
\put(3688,-803){\makebox(0,0)[lb]{\smash{{\SetFigFont{10}{12.0}{\rmdefault}{\mddefault}{\updefault}{\color[rgb]{0,0,0}$q$}%
}}}}
\put(3886,-242){\makebox(0,0)[lb]{\smash{{\SetFigFont{10}{12.0}{\rmdefault}{\mddefault}{\updefault}{\color[rgb]{0,0,0}$D_O$}%
}}}}
\put(4210,-1998){\makebox(0,0)[lb]{\smash{{\SetFigFont{10}{12.0}{\rmdefault}{\mddefault}{\updefault}{\color[rgb]{0,0,0}$o$}%
}}}}
\put(4563,-2682){\makebox(0,0)[lb]{\smash{{\SetFigFont{10}{12.0}{\rmdefault}{\mddefault}{\updefault}{\color[rgb]{0,0,0}$Q$}%
}}}}
\put(5358,-3617){\makebox(0,0)[lb]{\smash{{\SetFigFont{10}{12.0}{\rmdefault}{\mddefault}{\updefault}{\color[rgb]{0,0,0}$\ell_p$}%
}}}}
\put(3827,-3644){\makebox(0,0)[lb]{\smash{{\SetFigFont{10}{12.0}{\rmdefault}{\mddefault}{\updefault}{\color[rgb]{0,0,0}$z$}%
}}}}
\put(3978,-3087){\makebox(0,0)[lb]{\smash{{\SetFigFont{10}{12.0}{\rmdefault}{\mddefault}{\updefault}{\color[rgb]{0,0,0}$p$}%
}}}}
\put(2681,-3015){\makebox(0,0)[lb]{\smash{{\SetFigFont{10}{12.0}{\rmdefault}{\mddefault}{\updefault}{\color[rgb]{0,0,.56}$\gamma$}%
}}}}
\end{picture}%

%% file: TwoParts.pspdftex
\begin{picture}(0,0)%
\includegraphics{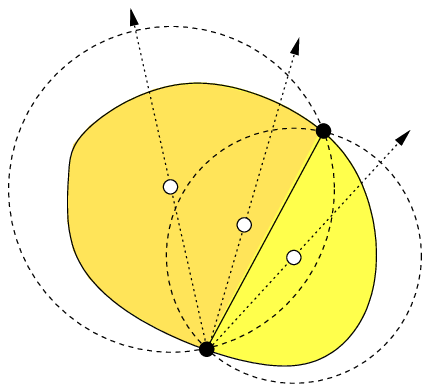}%
\end{picture}%
\setlength{\unitlength}{2368sp}%
\begingroup\makeatletter\ifx\SetFigFont\undefined%
\gdef\SetFigFont#1#2#3#4#5{%
  \reset@font\fontsize{#1}{#2pt}%
  \fontfamily{#3}\fontseries{#4}\fontshape{#5}%
  \selectfont}%
\fi\endgroup%
\begin{picture}(3724,3402)(1847,-3521)
\put(5410,-1312){\makebox(0,0)[lb]{\smash{{\SetFigFont{11}{13.2}{\rmdefault}{\mddefault}{\updefault}{\color[rgb]{0,0,0}$(u+5\alpha)[p]$}%
}}}}
\put(4416,-585){\makebox(0,0)[lb]{\smash{{\SetFigFont{11}{13.2}{\rmdefault}{\mddefault}{\updefault}{\color[rgb]{0,0,0}$u[p]$}%
}}}}
\put(2987,-350){\makebox(0,0)[lb]{\smash{{\SetFigFont{11}{13.2}{\rmdefault}{\mddefault}{\updefault}{\color[rgb]{0,0,0}$(u-5\alpha)[p]$}%
}}}}
\put(4711,-1218){\makebox(0,0)[lb]{\smash{{\SetFigFont{11}{13.2}{\rmdefault}{\mddefault}{\updefault}{\color[rgb]{0,0,0}$q$}%
}}}}
\put(3644,-3425){\makebox(0,0)[lb]{\smash{{\SetFigFont{11}{13.2}{\rmdefault}{\mddefault}{\updefault}{\color[rgb]{0,0,0}$p$}%
}}}}
\put(1862,-3344){\makebox(0,0)[lb]{\smash{{\SetFigFont{11}{13.2}{\rmdefault}{\mddefault}{\updefault}{\color[rgb]{0,0,0}$D_{pq}(u-5\alpha)$}%
}}}}
\put(5217,-3349){\makebox(0,0)[lb]{\smash{{\SetFigFont{11}{13.2}{\rmdefault}{\mddefault}{\updefault}{\color[rgb]{0,0,0}$D_{pq}(u+5\alpha)$}%
}}}}
\put(4111,-3025){\makebox(0,0)[lb]{\smash{{\SetFigFont{11}{13.2}{\rmdefault}{\mddefault}{\updefault}{\color[rgb]{0,0,0}$Q^+_{pq}(u)$}%
}}}}
\put(2824,-2312){\makebox(0,0)[lb]{\smash{{\SetFigFont{11}{13.2}{\rmdefault}{\mddefault}{\updefault}{\color[rgb]{0,0,0}$Q^-_{pq}(u)$}%
}}}}
\end{picture}%

%% file: qcdv1.pspdftex
\begin{picture}(0,0)%
\includegraphics{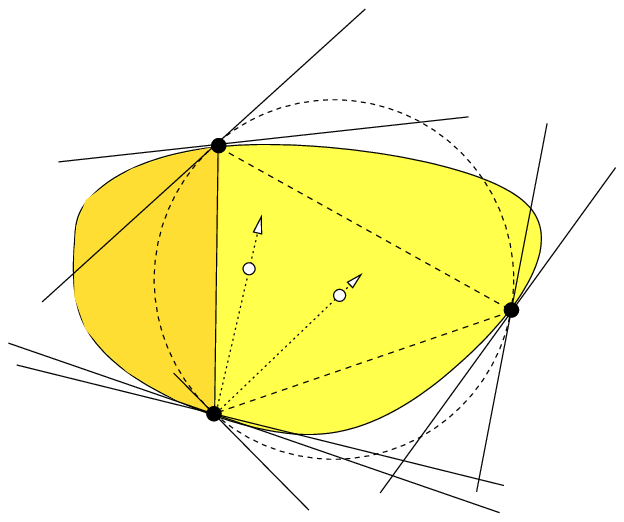}%
\end{picture}%
\setlength{\unitlength}{2368sp}%
\begingroup\makeatletter\ifx\SetFigFont\undefined%
\gdef\SetFigFont#1#2#3#4#5{%
  \reset@font\fontsize{#1}{#2pt}%
  \fontfamily{#3}\fontseries{#4}\fontshape{#5}%
  \selectfont}%
\fi\endgroup%
\begin{picture}(5137,4377)(1934,-4200)
\put(6946,-1486){\makebox(0,0)[lb]{\smash{{\SetFigFont{11}{13.2}{\rmdefault}{\mddefault}{\updefault}{\color[rgb]{0,0,0}$\ell_w$}%
}}}}
\put(6082,-3617){\makebox(0,0)[lb]{\smash{{\SetFigFont{11}{13.2}{\rmdefault}{\mddefault}{\updefault}{\color[rgb]{0,0,0}$\tau_p$}%
}}}}
\put(5956,-4104){\makebox(0,0)[lb]{\smash{{\SetFigFont{11}{13.2}{\rmdefault}{\mddefault}{\updefault}{\color[rgb]{0,0,0}$\ell_p$}%
}}}}
\put(5078,-54){\makebox(0,0)[lb]{\smash{{\SetFigFont{11}{13.2}{\rmdefault}{\mddefault}{\updefault}{\color[rgb]{0,0,0}$\tau_q^+$}%
}}}}
\put(4665,-531){\makebox(0,0)[lb]{\smash{{\SetFigFont{11}{13.2}{\rmdefault}{\mddefault}{\updefault}{\color[rgb]{0,0,0}$D_{pq}(u+5\alpha)$}%
}}}}
\put(6534,-976){\makebox(0,0)[lb]{\smash{{\SetFigFont{11}{13.2}{\rmdefault}{\mddefault}{\updefault}{\color[rgb]{0,0,0}$\tau_w^+$}%
}}}}
\put(6261,-2527){\makebox(0,0)[lb]{\smash{{\SetFigFont{11}{13.2}{\rmdefault}{\mddefault}{\updefault}{\color[rgb]{0,0,0}$w$}%
}}}}
\put(3688,-803){\makebox(0,0)[lb]{\smash{{\SetFigFont{11}{13.2}{\rmdefault}{\mddefault}{\updefault}{\color[rgb]{0,0,0}$q$}%
}}}}
\put(5903,-849){\makebox(0,0)[lb]{\smash{{\SetFigFont{11}{13.2}{\rmdefault}{\mddefault}{\updefault}{\color[rgb]{0,0,0}$\ell_q$}%
}}}}
\put(4325,-1962){\makebox(0,0)[lb]{\smash{{\SetFigFont{10}{12.0}{\rmdefault}{\mddefault}{\updefault}{\color[rgb]{0,0,0}$(u+5\alpha)$}%
}}}}
\put(3969,-1486){\makebox(0,0)[lb]{\smash{{\SetFigFont{11}{13.2}{\rmdefault}{\mddefault}{\updefault}{\color[rgb]{0,0,0}$u$}%
}}}}
\put(1949,-1490){\makebox(0,0)[lb]{\smash{{\SetFigFont{11}{13.2}{\rmdefault}{\mddefault}{\updefault}{\color[rgb]{0,0,0}$Q_{pq}(u)$}%
}}}}
\put(3644,-3388){\makebox(0,0)[lb]{\smash{{\SetFigFont{11}{13.2}{\rmdefault}{\mddefault}{\updefault}{\color[rgb]{0,0,0}$p$}%
}}}}
\put(4638,-4076){\makebox(0,0)[lb]{\smash{{\SetFigFont{11}{13.2}{\rmdefault}{\mddefault}{\updefault}{\color[rgb]{0,0,0}$\tau_p^+$}%
}}}}
\end{picture}%

%% file: CopyExists.pspdftex
\begin{picture}(0,0)%
\includegraphics{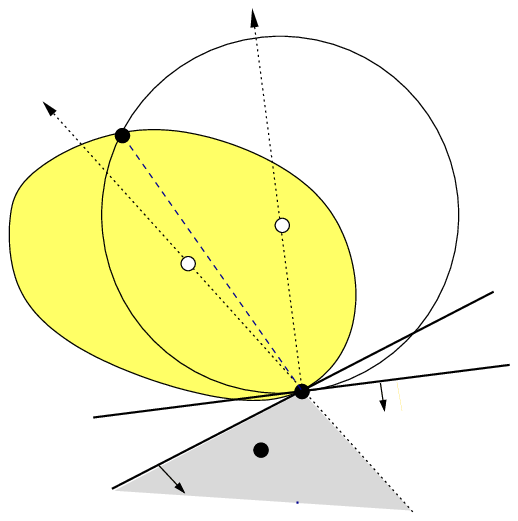}%
\end{picture}%
\setlength{\unitlength}{2368sp}%
\begingroup\makeatletter\ifx\SetFigFont\undefined%
\gdef\SetFigFont#1#2#3#4#5{%
  \reset@font\fontsize{#1}{#2pt}%
  \fontfamily{#3}\fontseries{#4}\fontshape{#5}%
  \selectfont}%
\fi\endgroup%
\begin{picture}(5181,4161)(382,-4201)
\put(3612,-807){\makebox(0,0)[lb]{\smash{{\SetFigFont{10}{12.0}{\rmdefault}{\mddefault}{\updefault}{\color[rgb]{0,0,0}$u[p]$}%
}}}}
\put(2131,-966){\makebox(0,0)[lb]{\smash{{\SetFigFont{10}{12.0}{\rmdefault}{\mddefault}{\updefault}{\color[rgb]{0,0,0}$q$}%
}}}}
\put(397,-2880){\makebox(0,0)[lb]{\smash{{\SetFigFont{10}{12.0}{\rmdefault}{\mddefault}{\updefault}{\color[rgb]{0,0,0}$Q_{pq}(u-5\alpha)$}%
}}}}
\put(3981,-1305){\makebox(0,0)[lb]{\smash{{\SetFigFont{10}{12.0}{\rmdefault}{\mddefault}{\updefault}{\color[rgb]{0,0,0}$D_{pq}(u)$}%
}}}}
\put(1148,-715){\makebox(0,0)[lb]{\smash{{\SetFigFont{10}{12.0}{\rmdefault}{\mddefault}{\updefault}{\color[rgb]{0,0,0}$(u-5\alpha)[p]$}%
}}}}
\put(3787,-3425){\makebox(0,0)[lb]{\smash{{\SetFigFont{10}{12.0}{\rmdefault}{\mddefault}{\updefault}{\color[rgb]{0,0,0}$p$}%
}}}}
\put(5548,-2918){\makebox(0,0)[lb]{\smash{{\SetFigFont{10}{12.0}{\rmdefault}{\mddefault}{\updefault}{\color[rgb]{0,0,0}$\tau_{p}(u)$}%
}}}}
\put(5426,-2310){\makebox(0,0)[lb]{\smash{{\SetFigFont{10}{12.0}{\rmdefault}{\mddefault}{\updefault}{\color[rgb]{0,0,0}$\ell_{p}(u-5\alpha)$}%
}}}}
\put(4569,-3446){\makebox(0,0)[lb]{\smash{{\SetFigFont{10}{12.0}{\rmdefault}{\mddefault}{\updefault}{\color[rgb]{0,0,0}$\mu_{p}(u)$}%
}}}}
\put(2550,-4110){\makebox(0,0)[lb]{\smash{{\SetFigFont{10}{12.0}{\rmdefault}{\mddefault}{\updefault}{\color[rgb]{0,0,0}$h_{p}(u-5\alpha)$}%
}}}}
\put(3645,-3781){\makebox(0,0)[lb]{\smash{{\SetFigFont{10}{12.0}{\rmdefault}{\mddefault}{\updefault}{\color[rgb]{0,0,0}$r$}%
}}}}
\end{picture}%

%% file: SecondDirection.pspdftex
\begin{picture}(0,0)%
\includegraphics{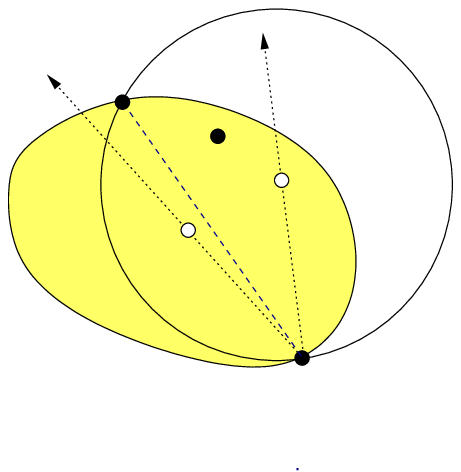}%
\end{picture}%
\setlength{\unitlength}{2368sp}%
\begingroup\makeatletter\ifx\SetFigFont\undefined%
\gdef\SetFigFont#1#2#3#4#5{%
  \reset@font\fontsize{#1}{#2pt}%
  \fontfamily{#3}\fontseries{#4}\fontshape{#5}%
  \selectfont}%
\fi\endgroup%
\begin{picture}(4682,3710)(374,-4027)
\put(3644,-3425){\makebox(0,0)[lb]{\smash{{\SetFigFont{10}{12.0}{\rmdefault}{\mddefault}{\updefault}{\color[rgb]{0,0,0}$p$}%
}}}}
\put(4849,-2627){\makebox(0,0)[lb]{\smash{{\SetFigFont{10}{12.0}{\rmdefault}{\mddefault}{\updefault}{\color[rgb]{0,0,0}$D_{pr}(u)$}%
}}}}
\put(2131,-966){\makebox(0,0)[lb]{\smash{{\SetFigFont{10}{12.0}{\rmdefault}{\mddefault}{\updefault}{\color[rgb]{0,0,0}$r$}%
}}}}
\put(3612,-807){\makebox(0,0)[lb]{\smash{{\SetFigFont{10}{12.0}{\rmdefault}{\mddefault}{\updefault}{\color[rgb]{0,0,0}$u[p]$}%
}}}}
\put(817,-3200){\makebox(0,0)[lb]{\smash{{\SetFigFont{10}{12.0}{\rmdefault}{\mddefault}{\updefault}{\color[rgb]{0,0,0}$Q_{pr}(u-5\alpha)$}%
}}}}
\put(389,-744){\makebox(0,0)[lb]{\smash{{\SetFigFont{10}{12.0}{\rmdefault}{\mddefault}{\updefault}{\color[rgb]{0,0,0}$(u-5\alpha)[p]$}%
}}}}
\put(3177,-1640){\makebox(0,0)[lb]{\smash{{\SetFigFont{10}{12.0}{\rmdefault}{\mddefault}{\updefault}{\color[rgb]{0,0,0}$q$}%
}}}}
\end{picture}%

%% file: collinearityNew.pspdftex
\begin{picture}(0,0)%
\includegraphics{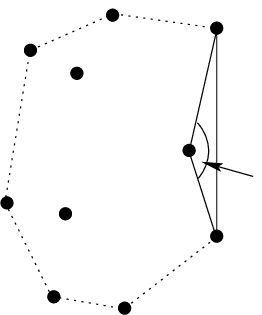}%
\end{picture}%
\setlength{\unitlength}{2368sp}%
\begingroup\makeatletter\ifx\SetFigFont\undefined%
\gdef\SetFigFont#1#2#3#4#5{%
  \reset@font\fontsize{#1}{#2pt}%
  \fontfamily{#3}\fontseries{#4}\fontshape{#5}%
  \selectfont}%
\fi\endgroup%
\begin{picture}(2067,2510)(148,-2903)
\put(1994,-2526){\makebox(0,0)[lb]{\smash{{\SetFigFont{10}{12.0}{\rmdefault}{\mddefault}{\updefault}{\color[rgb]{0,0,0}$b$}%
}}}}
\put(1573,-1654){\makebox(0,0)[rb]{\smash{{\SetFigFont{10}{12.0}{\rmdefault}{\mddefault}{\updefault}{\color[rgb]{0,0,0}$c$}%
}}}}
\put(2011,-600){\makebox(0,0)[lb]{\smash{{\SetFigFont{10}{12.0}{\rmdefault}{\mddefault}{\updefault}{\color[rgb]{0,0,0}$a$}%
}}}}
\put(2200,-1909){\makebox(0,0)[lb]{\smash{{\SetFigFont{10}{12.0}{\rmdefault}{\mddefault}{\updefault}{\color[rgb]{0,0,0}$\pi-\alpha$}%
}}}}
\end{picture}%

%% file: Grid.pspdftex
\begin{picture}(0,0)%
\includegraphics{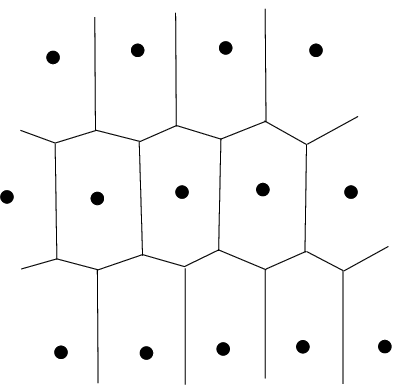}%
\end{picture}%
\setlength{\unitlength}{1973sp}%
\begingroup\makeatletter\ifx\SetFigFont\undefined%
\gdef\SetFigFont#1#2#3#4#5{%
  \reset@font\fontsize{#1}{#2pt}%
  \fontfamily{#3}\fontseries{#4}\fontshape{#5}%
  \selectfont}%
\fi\endgroup%
\begin{picture}(3764,3631)(7233,-4624)
\end{picture}%

%% file: Skeleton.pspdftex
\begin{picture}(0,0)%
\includegraphics{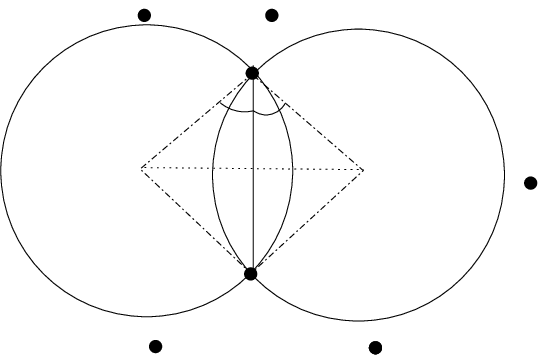}%
\end{picture}%
\setlength{\unitlength}{1973sp}%
\begingroup\makeatletter\ifx\SetFigFont\undefined%
\gdef\SetFigFont#1#2#3#4#5{%
  \reset@font\fontsize{#1}{#2pt}%
  \fontfamily{#3}\fontseries{#4}\fontshape{#5}%
  \selectfont}%
\fi\endgroup%
\begin{picture}(5165,3331)(1734,-4135)
\put(3034,-2255){\rotatebox{41.0}{\makebox(0,0)[lb]{\smash{{\SetFigFont{9}{10.8}{\rmdefault}{\mddefault}{\updefault}{\color[rgb]{0,0,0}$\frac{\beta}{2}d(p,q)$}%
}}}}}
\put(2712,-2558){\makebox(0,0)[lb]{\smash{{\SetFigFont{10}{12.0}{\rmdefault}{\mddefault}{\updefault}{\color[rgb]{0,0,0}$c_1$}%
}}}}
\put(5174,-2654){\makebox(0,0)[lb]{\smash{{\SetFigFont{10}{12.0}{\rmdefault}{\mddefault}{\updefault}{\color[rgb]{0,0,0}$c_2$}%
}}}}
\put(4011,-3779){\makebox(0,0)[lb]{\smash{{\SetFigFont{10}{12.0}{\rmdefault}{\mddefault}{\updefault}{\color[rgb]{0,0,0}$q$}%
}}}}
\put(3998,-1248){\makebox(0,0)[lb]{\smash{{\SetFigFont{10}{12.0}{\rmdefault}{\mddefault}{\updefault}{\color[rgb]{0,0,0}$p$}%
}}}}
\put(3858,-2088){\makebox(0,0)[lb]{\smash{{\SetFigFont{10}{12.0}{\rmdefault}{\mddefault}{\updefault}{\color[rgb]{0,0,0}$\theta$}%
}}}}
\put(4222,-2102){\makebox(0,0)[lb]{\smash{{\SetFigFont{10}{12.0}{\rmdefault}{\mddefault}{\updefault}{\color[rgb]{0,0,0}$\theta$}%
}}}}
\end{picture}%

%% file: norng.pspdftex
\begin{picture}(0,0)%
\includegraphics{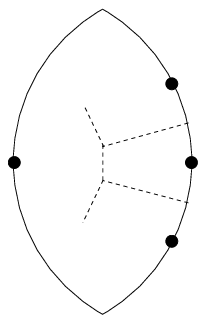}%
\end{picture}%
\setlength{\unitlength}{1657sp}%
\begingroup\makeatletter\ifx\SetFigFont\undefined%
\gdef\SetFigFont#1#2#3#4#5{%
  \reset@font\fontsize{#1}{#2pt}%
  \fontfamily{#3}\fontseries{#4}\fontshape{#5}%
  \selectfont}%
\fi\endgroup%
\begin{picture}(2370,3502)(3641,-4728)
\put(5996,-3000){\makebox(0,0)[lb]{\smash{{\SetFigFont{9}{10.8}{\rmdefault}{\mddefault}{\updefault}{\color[rgb]{0,0,0}$b$}%
}}}}
\put(5816,-4050){\makebox(0,0)[lb]{\smash{{\SetFigFont{9}{10.8}{\rmdefault}{\mddefault}{\updefault}{\color[rgb]{0,0,0}$d$}%
}}}}
\put(5801,-2010){\makebox(0,0)[lb]{\smash{{\SetFigFont{9}{10.8}{\rmdefault}{\mddefault}{\updefault}{\color[rgb]{0,0,0}$c$}%
}}}}
\put(3656,-2985){\makebox(0,0)[rb]{\smash{{\SetFigFont{9}{10.8}{\rmdefault}{\mddefault}{\updefault}{\color[rgb]{0,0,0}$a$}%
}}}}
\end{picture}%

%% file: wheel.pspdftex
\begin{picture}(0,0)%
\includegraphics{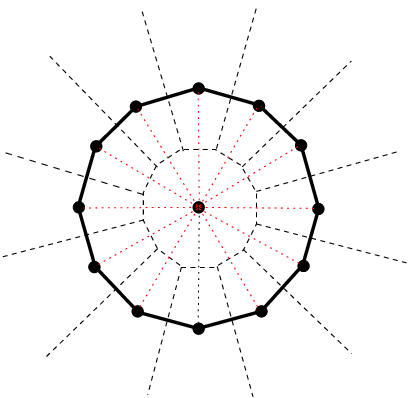}%
\end{picture}%
\setlength{\unitlength}{1973sp}%
\begingroup\makeatletter\ifx\SetFigFont\undefined%
\gdef\SetFigFont#1#2#3#4#5{%
  \reset@font\fontsize{#1}{#2pt}%
  \fontfamily{#3}\fontseries{#4}\fontshape{#5}%
  \selectfont}%
\fi\endgroup%
\begin{picture}(3914,3776)(2007,-4073)
\put(3985,-2315){\makebox(0,0)[lb]{\smash{{\SetFigFont{10}{12.0}{\rmdefault}{\mddefault}{\updefault}{\color[rgb]{0,0,0}$p$}%
}}}}
\end{picture}%